\documentclass[12pt]{article}
\usepackage[utf8]{inputenc}
\usepackage[T1]{fontenc}
\usepackage{titlesec}
\titleformat{\section}[block]{\large\normalfont\filcenter}{\thesection.}{.5em}{}
\titleformat{\subsection}[runin]{\normalfont\bfseries\filright}{\thesubsection.}{.5em}{}
\titleformat{\subsubsection}[runin]{\normalfont\bfseries\filright}{\thesubsubsection.}{.5em}{}

\usepackage{amsthm, amsmath,amsfonts,amssymb,verbatim,bm,xspace,mathtools}
\usepackage{setspace}
\usepackage[margin=1.25in]{geometry}
\usepackage{tikz}
\usetikzlibrary{calc,math, patterns, intersections, decorations.pathreplacing}
\usepackage{pgfplots}
\usepackage[justification=centering]{caption}
\usepackage{subcaption}
\usepackage{hhline}
\usepackage{float}
\usepackage{siunitx}
\pgfplotsset{compat=1.18}
\usepgfplotslibrary{fillbetween} 
\usepackage{epigraph}

\usepackage[colorlinks=true,linkcolor = blue, citecolor=blue, hyperfootnotes=false, hyperindex,breaklinks]{hyperref}
\usepackage[dvipsnames]{xcolor}
\usepackage{cleveref}

\usepackage{setspace}

\def\owhp{Decreasing Inequality Aversion\xspace}
\def\owhps{decreasing inequality aversion\xspace}
\def\owhm{Increasing Inequality Aversion\xspace}
\def\owhms{increasing inequality aversion\xspace}
\def\vc{\bm{c}}
\def\ve{\bm{e}}

\newcommand*\diff{\mathop{}\!\mathrm{d}}

\def\vunif{\lambda}

\def\Re{\mathbb{R}} \def\R{\mathbb{R}}

\def\ep{\varepsilon}

\def\al{\alpha}
\def\la{\lambda}

\def\phi{\varphi}

\def\zero{\mathbf{0}}
\def\one{\mathbf{1}}

\def\vx{\mathbf{x}}
\def\vy{\mathbf{y}}
\def\vz{\mathbf{z}}
\def\vpi{\bm{\pi}}

\def\os{\emptyset}

\newcommand{\df}[1]{\textit{#1}}

\newcommand{\norm}[1]{\| #1 \|}

\newcommand{\abs}[1]{ \left | #1 \right | }

\renewcommand{\df}[1]{\textbf{\textit{#1}}}

 \usepackage[multiple, flushmargin]{footmisc}

\newtheorem{theorem}{Theorem}

\newtheorem{lemma}{Lemma}
\newtheorem{proposition}{Proposition}

 \theoremstyle{definition}

\newtheorem{example}{Example}

\def\Claim{\textsc{Claim:} }

\makeatletter
\renewcommand\paragraph{\@startsection{paragraph}{4}{\z@}%
  {3.25ex \@plus1ex \@minus.2ex}
  {1.5ex \@plus .2ex}
  {\normalfont\bfseries}}            
\makeatother

\usepackage{quoting}
 \usepackage[shortlabels]{enumitem}
\usepackage[square, authoryear]{natbib}

\begin{document}
\title{Endogenous Inequality Aversion:\\ \Large Decision criteria for triage and other ethical trade-offs\thanks{ The authors thank Marcelo Gallardo for comments on an early draft of the paper.}}

\author{Federico Echenique\thanks{ \footnotesize University of California, Berkeley. \texttt{\footnotesize fede@econ.berkeley.edu}} \and Teddy Mekonnen\thanks{\footnotesize Brown University. \footnotesize\texttt{mekonnen@brown.edu}}  \and M.\ Bumin Yenmez\thanks{\footnotesize Washington University in St. Louis. Durham University Business School (UK).\texttt{\footnotesize bumin@wustl.edu}}}


\date{\today}
\maketitle

\raggedbottom

\begin{abstract}
Medical ``Crisis Standards of Care'' call for a utilitarian allocation of scarce resources in public emergencies, whereas standards of care under normal conditions place relatively greater priority on the worst-off. Inspired by such \textit{triage} rules, we study social welfare criteria whose distributive trade-offs depend on society's well-being, as captured by aggregate welfare. Because the welfare level determines the applicable aggregation criterion, while that criterion in turn determines welfare, the resulting criteria are self-referential. We provide an axiomatic foundation for a family of welfare criteria that become more utilitarian as aggregate welfare falls and more Rawlsian as it rises, thereby formalizing triage guidelines. We also characterize the converse case, in which priority to the worst-off increases as aggregate welfare falls.
\end{abstract}

 \begingroup
\allowdisplaybreaks

\clearpage 
\setcounter{page}{1}
\section{Introduction}
\bigskip

\setlength{\epigraphwidth}{.75\textwidth}
\epigraph{ ...at some times and locations, it will be necessary to change from ``conventional'' to ``contingency'' or ``crisis'' standards of medical care with a resulting change in triage approach from treating the ``sickest first'' to treating those ``most likely to survive'' first.
}{\textit{Triage and Treatment Tools for Use in a Scarce Resources-Crisis Standards of Care Setting After a Nuclear Detonation}\\\cite{coleman2011triage}
}
\bigskip

The allocation of scarce resources is a defining concern of economics and an everyday problem in medical practice. \textit{Triage} refers to the process of determining which patients receive access to which resources, and when. As described in the epigraph, triage decisions involve a trade-off between treating those ``most likely to survive''  and prioritizing the worst-off. Importantly, the resolution of this trade-off is context-dependent, with greater weight placed on efficiency during crises, such as pandemics, natural disasters, or the aftermath of a nuclear detonation. 

In this paper, we study such context-dependent welfare trade-offs. We take the relevant context to be society's overall welfare level, thereby keeping the context internal to the social welfare ordering itself. In the triage case, lower aggregate welfare calls for more efficient allocations, whereas higher aggregate welfare permits greater priority for the worst-off. We represent this shift through the welfare criterion used to aggregate the individual utilities induced by each allocation, with lower welfare associated with a more utilitarian criterion and higher welfare associated with a more Rawlsian one. 

This welfare-based context dependence, however, creates a feedback loop: the context determines the criterion used to aggregate individual utilities, while the aggregate welfare level generated by that criterion determines the context. We close this loop by imposing a consistency requirement: the welfare level used to determine the aggregation criterion must coincide with the welfare level generated when that criterion is applied. Aggregate welfare is thus both an input to and an output of the aggregation procedure, making the context endogenous and self-confirming. We axiomatize such a class of self-referential social welfare criteria and study their allocative implications. 

The remainder of the introduction proceeds in three steps. First, we explain why triage motivates welfare criteria that vary with context and why the relevant context should be defined at the population level. Second, we formalize this population-level context using aggregate welfare and show that doing so creates a consistency requirement. Third, we introduce our axioms and discuss the resulting class of self-referential social welfare criteria.

The ethical foundations of triage are reviewed by \cite{persad2009principles}, who group the main principles into favoring the worst-off patients, treating patients equally, maximizing total benefits, and promoting or rewarding social usefulness. Importantly, \cite{persad2009principles} note that the relative weight placed on these principles depends on the specific circumstances of triage. Under normal circumstances, worst-off patients are often prioritized. However, established triage practices during crises  adopt a more utilitarian approach.  \cite{wang2022fairness} writes that during a pandemic, ``it seems ethically justifiable to shift the focus from regular standards of care to crisis standards, where physicians' duty is primarily to improve and safeguard population health.''  Under such crisis standards of care, \cite{powell2008allocation} note that ``clinical guidelines propose both withholding and withdrawing ventilators from patients with the highest probability of mortality to benefit patients with the highest likelihood of survival.''\footnote{See also  \cite{white2009should}, \cite{rosenbaum2011ethical}, and  \cite{Christian1377}.} 

Such a shift toward utilitarianism was observed during the COVID-19 pandemic, as documented in two prominent case studies: New York State and Italy. New York was an early epicenter of the COVID-19 pandemic in the United States, and its pre-existing ventilator allocation guidelines became a point of national attention. Developed in 2015 by the New York State Task Force on Life and the Law, the guidelines were explicitly utilitarian, stating that ``The primary goal of the Guidelines is to save the most lives in an influenza pandemic where there are a limited number of available ventilators. To accomplish this goal, patients for whom ventilator therapy would most likely be lifesaving are prioritized.''\footnote{\url{https://nysba.org/wp-content/uploads/2020/05/2015-ventilator_guidelines-NYS-Task-Force-Life-and-Law.pdf}} Italy was one of the earliest countries to be severely affected by COVID-19. The Italian health-care system similarly emphasized a utilitarian objective in the allocation of ICU beds. As \cite{deceased2020siaarti} write, ``The underlying principle would be to save limited resources, which may become extremely scarce, for those who have a much greater probability of survival and life expectancy, in order to maximize the benefits for the largest number of people.''

To our knowledge, \citet{roemer2004eclectic} provides the earliest economic formalization of triage. He uses exogenously specified thresholds for individual welfare to identify regions of the utility space in which triage may override priority for the worse-off. These thresholds distinguish, for example, a ``tolerable death'' from a ``worthwhile life.'' Thus, the relevant context in Roemer's formulation is determined by the utility levels of each individual. 

In contrast, crisis standards of care motivate a society-wide rather than individual-level context. \citet{stroud2013crisis} describe crisis standards as marking ``the transition point at which resource allocation strategies focus on the community rather than the individual.'' We formalize this population-level perspective by taking society's aggregate welfare level as the relevant normative context.\footnote{We discuss alternative specifications of the relevant context in Section~\ref{sec:remarks}.} Doing so, however, creates a feedback loop, as we described earlier: aggregate welfare determines the applicable aggregation criterion, while that criterion, in turn, generates aggregate welfare. As the following example illustrates, this feedback loop may produce an inconsistent welfare evaluation.

\begin{example}\label{example:running}
Consider a planner who must allocate a single ventilator to one of two patients. Each patient $i$ has a type that determines her basic survival probability $p^0_i$. These probabilities can take on two values, $L$ (low) or $H$ (high), with $0<L<H$.  If $i$ is put on a ventilator, her survival probability rises to $p^1_i\coloneqq \gamma\cdot p^0_i$, for a fixed scalar $\gamma>1$. We may think of $\gamma$ as the effectiveness of ventilators, and we assume that $\gamma H \leq 1$. 

Assume that one of the patients is type $L$ while the other is type $H$. The utilitarian (efficient) criterion compares $\frac{1}{2}(\gamma L)+\frac{1}{2}H$ with $\frac{1}{2}(L)+\frac{1}{2}(\gamma H)$ and concludes that the ventilator must go to the patient of type $H$. This generates aggregate welfare level of $v^E\coloneqq \frac{1}{2}(L)+\frac{1}{2}(\gamma H)$. The Rawlsian criterion, instead, compares $\min\{\gamma L, H\}$ with $\min\{L, \gamma H\}$ and  concludes that the ventilator must go to the patient of type $L$. This generates an aggregate welfare of $v^F\coloneqq\min\{\gamma L, H\}$, where $v^F<v^E$. In other words, efficiency favors the better-off patient, while Rawlsianism favors the worse-off.

Now suppose that the aggregation criterion changes with the aggregate welfare level $v\in[0,1]$. Fix a threshold $\tau\in(0,1)$ such that the utilitarian criterion applies when $v<\tau$, corresponding to a crisis, while the Rawlsian criterion applies when $v\geq\tau$, corresponding to normal times. 

However, such an approach leads to inconsistent welfare evaluation if $\tau\in (v^F, v^E)$. If the planner applies the utilitarian criterion appropriate during a crisis, she allocates the ventilator to type $H$. Yet the resulting aggregate welfare of $v^E$ exceeds the threshold, placing society outside the crisis regime. Conversely, if the planner applies the Rawlsian criterion for normal times and allocates the ventilator to type $L$, the resulting aggregate welfare of $v^F$ falls below the threshold, placing society inside the crisis regime. Thus, \textit{neither aggregation criterion is consistent with the context generated by the allocation it selects.} More broadly, without a consistency requirement, the context used to justify a decision need not remain appropriate once the consequences of that decision are taken into account. 
\end{example}
\medskip

Our goal in this paper is to axiomatize a class of context-dependent and consistent social welfare criteria. Let us formalize our approach using the two-agent society of \autoref{example:running}, a special case of our general finite-agent model. Let $(x_1,x_2)\in\mathbb{R}_+^2$ denote a utility profile. Let $\Delta\coloneqq\{(\pi_1,\pi_2)\geq 0\mid \pi_1+\pi_2=1\}$ denote the simplex of welfare weights, with $\pi_i$ being a welfare weight on agent $i$. Given a nonempty set $\Pi\subseteq\Delta$, we evaluate $(x_1, x_2)$ according to the worst-case weighted utilitarian criterion
\[
\min_{(\pi_1, \pi_2)\in \Pi}\sum_{i=1}^2\pi_i x_i.
\]
The standard utilitarian criterion is obtained when $\Pi$ contains only the uniform weight $(1/2,1/2)$, while the Rawlsian criterion is obtained when $\Pi=\Delta$. More generally, each set $\Pi$ corresponds to a fixed welfare criterion. 

We are, however, not interested in a fixed criterion, but one that varies with the aggregate welfare level. Therefore, we allow the set of welfare weights to depend on the aggregate level of welfare itself. That is, if the aggregate level of welfare is $v$, then the set of relevant welfare weights is given by $\Pi(v)$. Finally, we impose a consistency requirement: for each utility profile $(x_1,x_2)$, its aggregate welfare level is a value $v$ satisfying 
\[
v=\min_{(\pi_1, \pi_2)\in \Pi(v)} \sum_{i=1}^2 \pi_i x_i.
\]
Hence, to know the relevant set of welfare weights, one needs to evaluate the social welfare function, but to calculate the value of this function, one must know the set of welfare weights to apply. Consequently, our proposed social welfare criteria involve a fixed-point property, making it self-referential.\footnote{The failure of consistency in \autoref{example:running} arises precisely due to a failure of such a fixed point when the relevant set of weights is given by $\Pi(v)=\{(1/2,1/2)\}$ for all $v<\tau$ and $\Pi(v)=\Delta$ for all $v\geq \tau$.}

Our main result is a representation of self-referential social welfare functions that move toward utilitarianism in times of crisis and Rawlsianism in times of abundance. We ground this representation in a new axiom, \textit{\owhms}, which states that if a society prefers an unequal outcome to an egalitarian one, this preference must hold as utilities are scaled down. Standard homothetic welfare criteria impose both this condition and its converse, making welfare trade-offs invariant to scale. In keeping with the ethical foundations for triage, we reject this view. We argue that inequality is a tolerable price for survival but an intolerable price in times of abundance.

We show that \owhms gives rise to a \emph{Fan social welfare function}; so-called because the iso-welfare curves ``fan in'' as welfare rises. This social welfare function is characterized by a set of welfare weights that is monotone in the set-inclusion order: $\Pi(v)\subseteq \Pi(v')$ whenever $v<v'$. At lower welfare levels, the smaller set of weights yields a comparatively more utilitarian criterion. As welfare rises, the set expands and the criterion gives progressively greater priority to the worst-off.

We apply our framework to a ventilator allocation problem within a population of agents and identify a tipping point: when the supply of ventilators relative to the share of the vulnerable population is above a threshold, our model directs resources to the most vulnerable patients. However, when this ratio falls below the threshold, either because ventilators become scarcer or the population's health deteriorates, the optimal policy switches to efficient allocation and prioritizes the strongest patients to maximize expected survival. Such a pivot mirrors the Crisis Standards of Care implemented in ICUs worldwide.

We also consider the opposite axiom, \textit{\owhps}, which states that if a society prefers an unequal outcome to an egalitarian one, this preference must hold when utilities are scaled up. Hence, \owhps captures a context-dependent welfare criterion in which more weight is placed on the worst-off agents as aggregate welfare declines.  We show that \owhps gives rise to another class of Fan social welfare functions that ``fan out'' as welfare rises: as $v$ improves, the set $\Pi(v)$ shrinks. Such social welfare functions exhibit the polar opposite dependence of the ethical criterion on welfare. In a crisis, there is an added concern for the more vulnerable agents, while higher levels of welfare lead to a greater emphasis on efficiency. We discuss applications to educational spending, documenting how governments have responded to crises by shifting resources to educational programs that favor vulnerable students. Another application is measuring income inequality, a literature that has traditionally discussed ways in which inequality aversion should depend on the overall level of welfare.

\section{Motivation: triage in ventilator allocation}\label{sec:motivation}
To motivate our model and main results, we build on the ventilator allocation example from the introduction. We sketch a simple, tractable model that exemplifies how our proposed welfare criterion works. A more formal analysis is presented in \Cref{sec:application}, following our main results.

Consider a unit mass of patients on the interval $[0,1]$. For some $\alpha\in (0,1)$, each agent $i\in [0,\alpha]$ is of type $L$, while each agent $i\in (\alpha, 1]$ is of type $H$. The planner has a mass $k\in (0,1)$  of ventilators.\footnote{We use a continuum population to simplify calculations and obtain closed-form expressions. Our main results in \Cref{sec:model} are for finite populations. We expect these results to extend to the continuum case.} We model a worsening crisis in two ways. First, since we have normalized the mass of patients to one, the parameter $k$ captures the scarcity of ventilators relative to demand. Lower values of $k$ correspond to situations where the existing ventilator supply is severely limited, as in a pandemic. Second, a worsening crisis is reflected in an increase in $\alpha$, capturing a larger fraction of vulnerable patients.

Each allocation policy induces a mapping $\vx:[0,1]\to [0,1]$ from patients to survival probabilities. If patient $i$ is not allocated a ventilator, we say that $i$ remains untreated and has a survival probability of $\vx_i=p_i^0$, where $p_i^0=L$ if $i\leq \alpha$ and $p_i^0=H$ if $i>\alpha$.  If patient $i$ receives a ventilator, we say the patient is treated and has a survival probability of $\vx_i=p^1_i=\gamma p^0_i$. Rather than focusing on the planner's allocation policy, we work directly with the vector (or function) $\vx$ of per-patient survival probabilities implied by the policy in question. Henceforth, we refer to $\vx$ itself as a policy.

The planner evaluates each policy $\vx$ by computing its worst-case weighted utilitarian welfare, where the relevant set of welfare weights depends on the aggregate welfare level of $\vx$ itself. Let $\Delta$ be the space of (Borel) probability measures on $[0,1]$. For a given welfare level $v$, let $\Pi(v)\subseteq \Delta$ denote the set of relevant welfare weights, and let the worst-case weighted utilitarian welfare of policy $\vx$ be defined by
\[
\inf \left\{ \int \vx(i)\diff \vpi(i) : \vpi\in \Pi(v) \right\}.
\]

To be concrete, here we focus on the correspondence $v\mapsto \Pi(v)$ defined by
\[
\Pi(v) = \left\{ (1-\rho(v))\vunif + \rho(v)\vpi \mid \vpi \in \Delta \right\},
\]
where $\vunif$ denotes the uniform measure and $\rho: [0,1]\to [0,1]$ is a smooth and 
strictly increasing function satisfying $\rho(0)=0$ and $\rho(1)=1$. The set of relevant weights thus expands symmetrically around $\vunif$ as $v$ increases---a \emph{$\rho$-contamination} around $\vunif$ (this terminology is borrowed from the robust Bayesian literature, where it is called an $\epsilon$-contamination). The parameter $\rho$ represents the \emph{degree of inequality aversion} at a given welfare level: as $v$ increases, $\Pi(v)$ expands, leading the planner to place more weight on the worse-off agents. In the limit, as $\rho(v) \to 1$, $\Pi(v)$ expands to the full simplex $\Delta$, and the planner’s objective approaches Rawlsianism. Conversely, as $\rho(v) \to 0$, the set $\Pi(v)$ collapses to the uniform measure, and the objective becomes utilitarian.

Given the specification of $\Pi$ as a $\rho$-contamination around $\vunif$ and some welfare level $v$, the worst-case weighted utilitarian welfare of policy $\vx$ can be simply expressed as
\[
\underbrace{(1-\rho(v))\bar{\vx} + \rho(v)\vx_{\min}}_{\eqqcolon F(v;\,\bar \vx,\, \vx_{\min})},
\]
where $\bar{\vx}=\int \vx(i)\diff \vunif(i)$ is the mean utility, and $\vx_{\min}$ is the utility of the worst-off patient, formally defined as the essential infimum of $\vx$ in $[0,1]$. Importantly, this worst-case weighted utilitarian welfare of $\vx$ is incomplete because we use $\Pi(v)$ for an arbitrary welfare level $v$. To complete our definition, $v$ must equal the actual aggregate welfare level of policy $\vx$, i.e., $v=u_\Pi(\vx)$, where $v$ is the fixed point
\begin{equation}
v = F(v;\,\bar \vx,\, \vx_{\min}).  
\end{equation}
The fixed point property is a crucial aspect of our proposal. The relevant level of inequality aversion is welfare-dependent, and welfare depends on the relevant level of inequality aversion.

Since there is a continuum of agents and only two possible types of agents, $\bar\vx$ and $\vx_{\min}$ take simple forms for all policies $\vx$. Furthermore, because $1>\bar\vx\geq \vx_{\min}>0$ by construction and $\rho$ is continuous and strictly increasing with $\rho(0)=0$ and $\rho(1)=1$, the above fixed point problem has a unique solution $v\in(0,1)$. Hence, the aggregate welfare level $u_\Pi(\vx)$ is well-defined for all policies $\vx$. 

For the illustrative purposes of this section, let us further assume that $\rho(v)=v$. In \Cref{sec:application} we consider the case of a general function $\rho$. The welfare level of $\vx$ is then 
\[
u_\Pi(\vx)=\frac{\bar\vx}{1+\bar\vx-\vx_{\min}}.
\]
Hence, given two policies $\vx$ and $\vy$, $u_\Pi(\vx)\geq u_\Pi(\vy)$ if and only if
\[
\bar\vx\left(1-\vy_{\min}\right)\geq \bar\vy\left(1-\vx_{\min}\right).
\]

We focus on two policies that seem natural for the ventilator allocation problem. An \emph{efficient} policy, denoted $\vx^E$,  treats as many type $H$ patients as possible because they generate a larger marginal gain, i.e., $(\gamma-1)H > (\gamma-1)L$. Any remaining ventilators are then used to treat as many type $L$ patients as possible. A \emph{fair} policy, denoted $\vx^F$, instead treats as many type $L$ patients as possible since they are the worst-off. Any remaining ventilators are then used to treat as many type $H$ patients as possible. In \Cref{lemma:wlog} of \Cref{sec:application}, we show that focusing on the fair and efficient policies is without loss of generality, as any other policy is dominated by one of these two.

We compare the adoption of the fair and the efficient policies under two scenarios. First, we consider a crisis. The mass of vulnerable patients is large, and ventilators are highly scarce. In particular, suppose that $k<\alpha$ so that it is impossible to treat all type $L$ patients. Thus, $\vx^E_{\min}=\vx^F_{\min}=L$. Since type $H$ patients generate a strictly larger gain, $\bar \vx^E>\bar \vx^F$. As a result, the planner should adopt a utilitarian criterion during a crisis and implement an efficient policy. 

Next, consider a scenario in which the mass of vulnerable patients is small and ventilators are more readily available. In particular, suppose $k>\max\{\alpha, 1-\alpha\}$, so that it is possible to treat all patients of either type. Under the efficient policy,  $\vx^E_{\min}=L$ and 
\[
\bar{\vx}^E =\underbrace{\gamma\Big(\alpha L + (1-\alpha) H\Big)}_{\text{Full treatment}} - (1-k)(\gamma - 1)L
\]
Under the fair policy, $\vx^F_{\min}=\min\{\gamma L,H\}$ and 
\[
\bar \vx^F=\text{Full treatment}-(1-k) (\gamma - 1)H.
\]
As $k\to 1$, both $\bar \vx^E$ and $\bar\vx^F$ converge to the same \emph{full treatment} value while $1-\vx^F_{\min}<1-\vx^E_{\min}$. As a result, the planner should adopt a Rawlsian criterion in the limit and implement a fair policy.

\begin{figure}[ht]
    \centering
\begin{tikzpicture}
    \begin{axis}[scale=1,
        xlabel={$k$ (Ventilator supply)},
        ylabel={$\alpha$ (Population vulnerability)},
        xmin=0, xmax=1,
        ymin=0, ymax=1,
        xtick={0, 1},
        xticklabels={0, 1},
        ytick={1},
        yticklabels={1},
        axis lines=left,
        title={Optimal Policy Regions},
        width=9cm,
        height=9cm,
        area style,
        axis on top,
    ]

    \draw[blue, ultra thick, name path=curve] plot[domain=0:1/7]( \x,{\x})--plot[domain=1/7:7/9](\x, {(1+\x)/8})--plot[domain=7/9:27/29]   ( \x, {(37*\x-27)/8})--plot[domain=27/29:1]   (\x,{\x}); 

 \addplot[name path=capacity, black, thick, dashed] coordinates {(0,0) (1,1)};

    \addplot[name path=left, draw=none] coordinates {(0,1)};
    \addplot[name path=right, draw=none] coordinates {(1,0)};


    \addplot[green!15] fill between[of=curve and right];

   \addplot[orange!15] fill between[of=curve and capacity];

  \addplot[red!15] fill between[of=capacity and left];


    \node[align=center, text=green!30!black, scale=.9] at (axis cs: 0.6, .1) { 
       \footnotesize \textsc{Fair policy}\\
        \tiny (Prioritize $L$)
    };

    \node[align=center, text=orange!30!black, scale=.9] at (axis cs: 0.6, .33) { 
        \footnotesize \textsc{Efficient policy}\\
        \tiny (Prioritize $H$)
    };

    \node[align=center, text=red!30!black, scale=.9] at (axis cs: 0.4, .8) {
        \footnotesize \textsc{Efficient policy}\\
        \tiny (Prioritize $H$)\\
        \tiny $\alpha > k$ (crisis region)
    };

    \node[anchor=east, blue, scale=.9] at (axis cs: 1, .55) {\footnotesize $\alpha^*(k)$};
    
    \end{axis}
\end{tikzpicture}
 \captionsetup{oneside,margin={0cm,0cm},justification=justified, singlelinecheck=false}
    \caption{Optimal policy depending on fraction of vulnerable (Type $L$) patients and ventilator supply. The blue line illustrates the threshold $\alpha^*(k)$. In this case, $\alpha^*(k)=\frac{1+k}{8}$  when $k\in [1/7, 7/9]$; $\alpha^*(k)=\frac{37k-28}{8}$ when  $k\in [7/9, 27/29]$; and $\alpha^*(k)=k$ otherwise.}
    \label{fig:ventilators}
\end{figure}
Broadly, whether the efficient or the fair policy is optimal depends on the population's vulnerability relative to the scarcity of ventilators. As detailed in \Cref{sec:application}, \autoref{prop:optimal threshold} establishes that for any ventilator supply $k\in(0,1)$, there exists a threshold $\alpha^*(k)$ such that $\vx^E$ is optimal  when $\alpha>\alpha^*(k)$ (i.e., when the share of vulnerable agents is large) and $\vx^F$ is optimal when $\alpha<\alpha^*(k)$. Moreover, because $\alpha^*(k)$ is increasing in $k$, the planner is more likely to prioritize efficiency over the worst-off agents as ventilators become scarcer.

\autoref{fig:ventilators} illustrates these comparative statics for the case where $\rho(v)=v$, $H=0.5$, $L=0.1$, and $\gamma=2$. Holding supply fixed, the optimal policy shifts from the prioritization of the worst-off to efficiency as the population becomes more vulnerable (as $\alpha$ increases). Conversely, fixing the share of vulnerable agents, the planner shifts toward efficiency as the supply of ventilators shrinks (as $k$ decreases). Because $\alpha^*(k)$ is increasing, a larger supply of ventilators allows the planner to prioritize the worst-off agents even at higher levels of population vulnerability. While the exact threshold depends on model primitives, when $\alpha>k$, we reach a critical level beyond which switching to the efficient policy is optimal regardless of the parameters of the model. We term the latter region the ``crisis region.'' 

\section{The Model}\label{sec:model}

Consider a society consisting of $n$ agents. A planner must make a policy decision that affects everyone in society, such as the allocation of scarce resources. Each agent has a utility function defined over possible collective decisions. Each decision thus maps into a vector of ``utils'' $\vx\in\Re^n$ describing the level of utility achieved by each of the $n$ agents. We treat these utils as interpersonally comparable: implicit in our model is the notion that utility has some objective basis. In particular, in medical ethics, utility is often measured in life-years or survival probabilities; and the medical practice of triage is predicated on the comparison of such measures of individual welfare. 

For our purposes, we abstract from the details of the planner's decision problem and work directly with the utility vectors induced by the planner's possible decisions. Our model is therefore \df{welfarist.} Our work seeks to rank utility vectors through a social welfare function that captures the desirability of any potential collective decision. In the triage application we have discussed in the introduction, the planner's decision to allocate a ventilator to patient 1 maps into the vector $(\gamma L, H)$ while her decision to allocate the ventilator to patient 2 maps into $(L, \gamma H)$. 

The literature on social welfare functions is, of course, vast (a partial exposition may be found in \cite{moulin_book}). Two classical examples of social welfare orderings are the utilitarian and Rawlsian orderings. The utilitarian ordering ranks $\vx$ over $\vy$ if and only if $\frac{1}{n}\sum_i \vx_i\geq \frac{1}{n}\sum_i \vy_i$, while the Rawlsian ``max-min'' ordering requires that $\min\{\vx_i\}\geq \min\{ \vy_i\}$. These two welfare orderings are a special case of weighted utilitarianism, $\sum_i\vpi_i \vx_i$, where $\vpi_i$ is a \df{welfare weight} attached to the utility of agent $i$. We assume that $\sum_i\vpi_i=1$. These welfare aggregation methods are the subject of very lengthy literatures in ethics and social choice.\footnote{For an axiomatic treatment of these rules, see, for example, \cite{d1977equity} or \cite{harsanyi1955cardinal} (and \cite{zhou1997harsanyi} for a rigorous and very general version of Harsanyi's result). Our axiomatization is, however, very different from these.} 

\subsection{Preliminary definitions and notation} 
Let $\Re^n_+$ represent the set of nonnegative real vectors, and $\Re^n_{++}$ the set of positive real vectors. Let $\zero$ and $\one$ denote the vectors of zeroes and ones in $\R^n$, respectively. For any two vectors $\vx,\vy\in \Re^n$, $\vx\geq \vy$ if $\vx_i\geq \vy_i$ for every $i \in \{1,\ldots,n\}$. Additionally, $\vx>\vy$ if $\vx\geq \vy$ and  $\vx_i>\vy_i$ for some $i\in \{1,\ldots,n\}$. Similarly, $\vx \gg \vy$ if $\vx_i>\vy_i$ for every $i\in \{1,\ldots,n\}$. Thus $\Re^n_{++}=\{\vx\in\Re^n \mid \vx \gg \zero \}$. We write the inner product between two vectors as $\vx\cdot \vy = \sum_i \vx_i \vy_i$. When $A\subseteq\Re^n$ is a set and $\vx\in\Re^n$ is a vector, we write $\vx\cdot A$ for $\{\vx\cdot \vy:\vy\in A\}$. 

The $(n-1)$-dimensional simplex in $\Re^n$ is denoted by $\Delta \coloneqq \{\vx\in \Re^n_+ \mid \sum_{i=1}^n \vx_i=1\}$.

As a mathematical convention, we assume that $\inf \emptyset = +\infty$. 

We omit the (standard) definitions of continuity and upper hemicontinuity of correspondences. See, for example, \cite{klein1984theory} or \cite{aliprantis2006infinite}.

\subsection{Axioms.}

In our model, the elements of $\Re^n_+$ are utility vectors. Given a scalar $c \in \Re_+$, we denote the corresponding constant vector $(c,\dots, c)\in\Re^n_+$ by $\vc$; thus $\vc = c\one$. These constant vectors represent an egalitarian outcome in which all agents enjoy the same level of utility. Several of our axioms rely on comparisons with constant vectors because they restrict deviations from a preference for egalitarian outcomes.
 
Let $\succeq$ be a complete and transitive binary relation on $\Re^n_+$; these are called weak orders or rational preferences. The binary relation $\succeq$ represents a systematic ranking over utility vectors: an ordinal welfare ranking.

Consider the following properties, or axioms, of the welfare ranking $\succeq$:
\begin{enumerate}
\item \textit{\owhm :} \\ For any $c\in \Re_+$ and $\la \in (0,1)$, $\vx\succ \vc$ implies that $\la \vx\succ \la \vc$.
\item \textit{\owhp :} \\ For any $c\in \Re_+$ and $\la \in (1,\infty)$, $\vx\succ \vc$ implies that $\la \vx\succ \la \vc$. 
\end{enumerate}

\owhm and \owhp are the main substantive axioms in our model, and they are crucial for the social welfare functions that we propose. \owhm says that if the potentially non-egalitarian outcome $\vx$ is deemed preferable to a fully egalitarian outcome $\vc$, then the comparison cannot be overturned when we shrink $\vx$ and $\vc$ proportionally. If we are comfortable giving up the egalitarian outcome $\vc$ in favor of outcome $\vx$, then we must be comfortable doing so when comparing these outcomes at a proportionally lower overall level of well-being.

\owhp is the flip side of \owhms. If $\vx$ is preferable to a fully egalitarian outcome $\vc$, then the preference must be preserved when we blow up $\vx$ and $\vc$ proportionally.  \owhp captures the idea that fairness concerns are more important in bad times than in good times. We shall see that it provides, in a sense, the opposite kind of social welfare than \owhms.

Imposing both \owhms and \owhps is the same as assuming \textit{homotheticity}. For example, the Nash social welfare function $\prod_{i=1}^n \vx_i^{a_i}$ (with parameters $(a_1,\dots, a_n)\gg 0$) is homothetic and therefore displays what we would call constant inequality aversion. 

In addition to increasing and \owhps, we shall consider six additional properties of $\succeq$. The first five are either standard or variations of standard axioms.
\begin{enumerate}
\item \textit{Monotonicity:} For any $\vx,\vy \in \Re^n_+$, $\vx\geq \vy$ implies $\vx\succeq \vy$, and $\vx\gg \vy$ implies $\vx\succ \vy$.
\item \textit{Inada:} If $\vx\in\Re^n_+\setminus\Re^n_{++}$, then $\vx\sim \zero$. 
\item \textit{Weak continuity:} $\{\vx\in\Re^n_+:\vx\succ \vc\}$ is open for any $c\in \Re_+$; and the sets 
$\{c\in\Re_+:\vc \succeq \vx \}$ and $\{c\in\Re_+ \mid \vx\succeq \vc\}$
are closed for any $\vx\in\Re^n_+$.
\item \textit{Strong continuity:} The upper contour set $\{\vy\in\Re^n_+:\vy\succeq \vx\}$ is closed for any $\vx\in \Re^n_+$; and  the lower contour set $\{\vy\in\Re^n_+:\vx \succeq \vy \}$ is closed for any $\vx\in\Re^n_+$.
\item \textit{Convexity:} For any $c\in \Re_+$, $\{\vx\in\Re^n_+ \mid \vx\succ \vc \}$ is convex. 
\item\label{ax:flat} \textit{Mixing invariance:} For any $\vx \in \Re^n_+$, $c\in \Re_+$, and  $\al\in (0,1)$,  $\vx\sim \vc$ implies $\al \vx + (1-\al) \vc\sim \vc$.
\end{enumerate}

The first two axioms are straightforward. Monotonicity means that the planner respects the Pareto principle. Inada singles out zero utility as being especially unacceptable: If one agent has zero utility, it is as if all agents had zero utility. As we shall see, Inada is required to characterize the representation with a monotone decreasing fan. It is not used to characterize the class of social welfare functions applicable to triage problems. 

The two continuity axioms are ``technical,'' but their role in our results is more interesting than in typical representation theorems. Strong continuity is the standard axiom imposed to ensure the existence of a utility representation. We use strong continuity to obtain a representation with a monotone increasing fan. For a monotone decreasing fan, strong continuity turns out to be too strong. We instead impose weak continuity so that we can accommodate situations like one of the examples in \Cref{ex:buminsex} (see \autoref{fig:buminexc}). We discuss this weakening in detail when describing \autoref{fig:buminex}.

The subsequent axioms are more substantive. In these axioms, the utility vectors that are equal in all components (i.e., constant vectors) play a special role. Essentially, a constant vector $\vc\in\Re^n$ gives utility $c\in\Re$ for each agent $i$. So it embodies an egalitarian ideal in which all agents achieve the same level of well-being.

Convexity says that if we regard two potentially non-egalitarian outcomes, $\vx$ and $\vy$, as more desirable than an egalitarian outcome $\vc$, then by mixing $\vx$ and $\vy$ in a convex combination $\al \vx + (1-\al)\vy$, we obtain an outcome that is still more preferable than $\vc$. The idea is that the mixture $\al \vx + (1-\al)\vy$ results in less unequal outcomes relative, at least, to the most unequal among $\vx$ and $\vy$. Therefore, if the deviation from egalitarianism in $\vx$ and $\vy$ is acceptable, at least in comparison to $\vc$, then so is the presumably smaller deviation from egalitarianism in the convex combination. As an axiom, some version of convexity (often a stronger version than we have assumed here) is extremely common in demand theory, decision theory, and social choice. 

Our sixth axiom, mixing invariance, is a qualification of convexity. It states that when a potentially non-egalitarian outcome $\vx$ is regarded as exactly indifferent to a fully egalitarian outcome $\vc$, then there is no strict benefit in a convex combination of $\vx$ and $\vc$. Any \emph{strict} benefit in mixing $\vx$  must result from mixing non-egalitarian outcomes, not from mixing $\vx$ with an egalitarian outcome that is viewed as equivalent to $\vx$. For example, consider a version of the ventilator allocation problem that we discussed in \Cref{sec:motivation}. There are many patients who are classified into a small number of types. Suppose that $\vx$ results from unequal treatment of different types of agents while $\vc$ treats everyone the same. Despite the inequality in $\vx$, the planner regards these as equally desirable: $\vx\sim \vc$. Now suppose that half the agents in each type are given the outcome behind $\vx$ and half are given the outcome behind $\vc$, resulting in the vector $\vy$. Mixing invariance requires that $\vy\sim \vx\sim \vc$. We may think of this requirement as a separability criterion that applies whenever a policy is indifferent to an egalitarian outcome.

Mixing invariance generates the piecewise linear iso-welfare curves in our social welfare functions. Without mixing invariance, we would obtain a representation with the qualitative features of our Fan social welfare functions, but without its simple and tractable closed form. Mixing invariance means that we only need to track the set of welfare weights at the equal-welfare vectors $\vc$. If we didn't impose it, we would need to keep track of the relevant set of welfare weights locally for each $\vx$. We don't have a simple way of doing this.

\subsection{The ``Fan'' social welfare function}

We are interested in social welfare functions that evaluate a utility profile $\vx\in\Re_+^n$ by computing $\min\{\vpi\cdot \vx:\vpi\in \Pi\}$ for some convex and compact subset $\Pi\subseteq\Delta$ of welfare weights. In the utilitarian case, the only relevant weight is the uniform weight, i.e., $\Pi=\{(1/n,\ldots,1/n)\}$. The Rawlsian rule results from allowing all weights, $\Pi=\Delta$. Intermediate cases correspond to sets that are neither singletons nor the full simplex. We can thus think of the size of the set $\Pi$ as the degree of aversion to inequality, similar to the role played by $\rho$ in the application to ventilator allocation from Sections~\ref{sec:motivation} and~\ref{sec:application}. Specifically, given two sets $\Pi$ and $\Pi'$ with $\Pi\subseteq \Pi'$, the planner is more concerned for the worst-off agent under $\Pi'$ than under $\Pi$. One may think of the set $\Pi$ as controlling the curvature of the social welfare function; the iso-welfare curves associated with $\Pi'$ exhibit larger curvature than those of $\Pi$.

We now consider sets of welfare weights $\Pi(c)\subseteq \Delta$ parameterized by a scalar $c\in\Re_+$, which represents the relevant context. While we discuss different specifications of context in Section~\ref{sec:remarks}, here, the context is the level of welfare. Thus, $\Pi(c)$ is the relevant set of weights when the level of welfare is $c$.\footnote{The use of such sets of weights is arguably consistent with Rawls' own discussion of the criteria he proposes in \cite{Rawls1971} and \cite{rawlsDJ}.} For this to make sense, $c$ must equal the actual welfare level: a self-referential property. In other words, the social welfare level associated with a utility profile $\vx$ is obtained by applying the relevant set of  welfare weights, where the set of relevant weights itself depends on the welfare level associated with $\vx$. The representation is therefore the outcome of a fixed-point argument. 

Formally, we define the social welfare of a utility profile $\vx$ as
\[
u_{\Pi}(\vx) \coloneqq \inf \{v \mid \min\{\vpi\cdot \vx \mid \vpi\in \Pi(v) \} = v \}.
\]
For the classes of correspondences considered in our representation theorems, the infimum is achieved at $v=u_{\Pi}(\vx)$ (see Section~\ref{sec:proofs}). Hence, we say that $\vx$ is ranked above $\vy$ when 
\[ u_{\Pi}(\vx)=
\min\{\vpi\cdot \vx \mid \vpi\in \Pi(u_{\Pi}(\vx)) \} \geq \min\{\vpi\cdot \vy \mid \vpi\in \Pi(u_{\Pi}(\vy)) \}=u_{\Pi}(\vy),
\] so that the social welfare function has a max-min aspect, where the relevant set of weights used in the ``min'' part of the representation endogenously depends on the overall level of welfare itself---$u_{\Pi}(\vx)$ or $u_{\Pi}(\vy)$. 

The displayed inequality illustrates that $u_{\Pi}(\vx)$ appears on both sides of the equal sign on the left-hand side of the inequality, and same for $u_{\Pi}(\vy)$ on both sides of the equal sign on the right-hand side. Hence, the representation requires a fixed-point argument (which is a novelty in representation theorems for social orderings).\footnote{A precedent in decision theory is \cite{gul1991theory}. Recursive preferences such as \cite{epstein1991substitution} also have a circular component, albeit of a different nature than ours. The fixed-point property also explains the role of the outer ``inf,'' which plays a role analogous to that of the infimum in the proof of Tarski's fixed-point theorem. We use similar reasoning to obtain the existence of a fixed point even in the absence of continuity. A similar argument using ``sup'' instead is possible.} 

We say that the correspondence $c\mapsto \Pi(c)$ is a \df{monotone increasing fan} if
\begin{enumerate}
\item for any $c\in\Re_+$, $\Pi(c)$ is nonempty, closed, and convex, and 
\item for any $c,c'\in\Re_+$, $c\leq c'$ implies $\Pi(c)\subseteq \Pi(c')$.
\end{enumerate} The name is chosen to reflect iso-welfare curves that ``fan in'' as welfare increases.

If we replace the second property with $c\leq c'$ implies $\Pi(c)\supseteq \Pi(c')$, then we say that $\Pi(\cdot)$ is a \df{monotone decreasing fan}, where the name is chosen to reflect iso-welfare curves that ``fan out'' as welfare increases (equivalently, ``fan in'' as welfare decreases).

In a representation with a monotone increasing fan, when the level of well-being $u_{\Pi}(\vx)$ is low, the set of weights is smaller than when $u_{\Pi}(\vx)$ is high. Hence, the concern for the lowest levels of utility in $\vx$ is greater when $u_{\Pi}(\vx)$ is high than when it is low. Such a representation captures the criteria for triage in which low levels of social welfare correspond to a greater degree of utilitarianism than in high levels of social welfare. 

In contrast, a monotone decreasing fan entails the opposite comparative statics. Low levels of social welfare correspond to a more Rawlsian criterion than high levels of welfare. As we shall see, these two representations also differ in some technical aspects. They involve different continuity axioms, treat zero welfare in different ways, and rely on substantially different fixed-point arguments.

\subsection{An example}\label{ex:buminsex} 

\autoref{fig:buminex} illustrates the iso-welfare curves of monotone increasing and monotone decreasing social welfare functions when $n=2$. \autoref{fig:buminexa} depicts a social welfare function with a monotone increasing fan. In the lower-left region, the iso-welfare curves are linear (utilitarian). As welfare levels rise, the curves become increasingly kinked, eventually resembling the Leontief curves of a Rawlsian social welfare function. Such iso-welfare curves capture, for example, the shift in triage guidelines between normal times and crises that we discussed in the introduction.

\begin{figure}[ht]
    \centering
    \begin{center}
\begin{subfigure}[t]{0.32\textwidth}
\begin{tikzpicture}[scale=.65]
  \draw[->, thick] (-.05,0) -- (7,0) node[below] {\footnotesize $\vx_1$};
  \draw[->, thick] (0,-.05) -- (0,7) node[left] {\footnotesize $\vx_2$};

  \draw[dashed, gray] (0,0) -- (6.5,6.5);

  \foreach \x in {1.5, 2, 2.5, 3, 3.5, 4, 4.5, 5, 5.5, 6, 6.5, 7}{
    \draw[thick, blue!80] (0.1, \x - 0.1) -- ({\x/2 - \x*\x/30}, {\x/2 - \x*\x/30}) -- (\x - 0.1, 0.1);
  }

\draw[thick, blue!80] (.9, 7) -- (2,2) -- (7,.9);
\draw[thick, blue!80] (1.1, 7) -- (2.1,2.1) -- (7,1.1);
\draw[thick, blue!80] (1.5, 7) -- (2.3,2.3) -- (7,1.5);

  \foreach \k in {2.5, 3.5, 4.5}{
    \draw[thick, blue!80] (6.5, \k) -- (\k, \k) -- (\k, 6.5);
  }
\end{tikzpicture}
\caption{Monotone increasing}
\label{fig:buminexa}
\end{subfigure}
\begin{subfigure}[t]{0.32\textwidth}
\begin{tikzpicture}[scale=.65]
  \draw[->, thick] (-.05,0) -- (7,0) node[below] {\footnotesize $\vx_1$};
  \draw[->, thick] (0,-.05) -- (0,6.5) node[left] {\footnotesize $\vx_2$};

  \foreach \x in {1.5,2,2.5,3,3.5}{
    \draw[thick,blue] (.1,\x-.1) -- (\x/2-.5,\x/2-.5) -- (\x-.1,.1);
  }

  \foreach \x in {3.7,4,4.5,5}{
    \draw[thick,blue] (.1,\x-.1) -- (\x/2-.35,\x/2-.35) -- (\x-.1,.1);
  }

  \foreach \x in {5.5,6}{
    \draw[thick,blue] (.1,\x-.1) -- (\x-.1,.1);
  }
\end{tikzpicture}
\caption{Monotone decreasing}
\label{fig:buminexb}
\end{subfigure}
\begin{subfigure}[t]{0.32\textwidth}
\begin{tikzpicture}[scale=.65]
  \draw[->, thick] (-.05,0) -- (7,0) node[below] {\footnotesize $\vx_1$};
  \draw[->, thick] (0,-.05) -- (0,6.5) node[left] {\footnotesize $\vx_2$};

  \foreach \x in {0.5,1,1.4}{
    \draw[thick,blue] (\x,6) -- (\x,\x) -- (7,\x);
  }
  \foreach \x in {3,3.5,4,4.5,5,5.5,6,6.5,7,7.5}{
    \draw[thick,blue] (1.5,\x-1.5) -- (\x-1.5,1.5);
  }
    \filldraw [fill=RedOrange, draw=black] (3,3) circle (0.05);
  \draw [dotted,thin,red](3,-.02)node[below=3pt]{$c$}--(3, 3);
    \draw [dotted,thin,red](-.02, 3)node[left=3pt]{$c$}--(3, 3);
  \draw [dotted](1.4,-.02)node[below]{$c^*$}--(1.4, 1.4);
    \draw [dotted](-.02, 1.4)node[left]{$c^*$}--(1.4, 1.4);
    \filldraw [fill=black, draw=black] (1.4,1.4) circle (0.05);
\end{tikzpicture}
\caption{Monotone decreasing}
\label{fig:buminexc}
\end{subfigure}
 \captionsetup{oneside,margin={0cm,0cm},justification=justified, singlelinecheck=false}
    \caption{Three examples of ``fan'' social welfare orderings. The examples are symmetric for illustration; symmetry is not imposed in the representation theorems.}
    \label{fig:buminex}
    \end{center}
\end{figure}

The other two examples in \autoref{fig:buminex} show representations with decreasing fans. \autoref{fig:buminexb} represents the converse of \autoref{fig:buminexa}: here, the iso-welfare curves in the lower-left resemble Leontief indifference curves. As the welfare level rises, the curves flatten, eventually becoming linear and thus utilitarian. Such iso-welfare curves could, for example, capture shifts in vaccination priorities. During the early stages of a pandemic, when immunity is low, the criterion would prioritize vaccinating the most vulnerable: older populations, overexposed essential workers, and agents with pre-existing comorbidities. In later stages, when immunity is high, the criterion shifts toward prioritizing efficient vaccination schedules.

\autoref{fig:buminexc} depicts a particularly striking example of a monotone decreasing fan. The social welfare function is exactly Rawlsian until the worst-off agent achieves a threshold utility $c^*$; thereafter, the social welfare function becomes exactly utilitarian. This is depicted by Leontief iso-welfare curves, up until all agents achieve a utility of $c^*$. Then the iso-welfare curves become linear and correspond to a standard utilitarian objective. Formally, the decreasing fan is given by the correspondence
\[
\Pi(c) = \begin{cases}
  \Delta & \text{ if } c\leq c^* \\
  \{(1/n,\ldots,1/n) \} & \text{ if } c> c^*.
  \end{cases}
\]
Note that $\Pi$ is upper hemicontinuous in addition to being monotonically decreasing. Indeed, the only issue in verifying upper hemicontinuity is  a sequence $c_n\downarrow c^*$; but since $\Pi(c^*)\supseteq \Pi(c_n)$, the closed graph property holds immediately. 

The example in \autoref{fig:buminexc} illustrates the need to relax the standard strong continuity axiom, which requires weak upper contour sets to be closed.  Consider a sequence  
\[
\vx^k=\left(\frac{c+kc^*}{1+k}, \frac{c(2k+1)-kc^*}{1+k}\right)
\]
with $c>c^*$. Notice that $\min\{\vx^k_i\}>c^*$ and $\frac{1}{n}\sum_ix_i^k=c$ for all $k$, and observe that $\vx^0$ corresponds to the red point at $(c,c)$ in \autoref{fig:buminexc}. As $k$ increases, we move northwest away from this red point along the (linear) iso-welfare curve that passes through $(c,c)$. However, the sequence converges to a vector $\vx^\infty$ with $\min\{\vx_i^\infty\}=c^*$, which lies outside the upper contour set at $c$. Thus, the strong continuity axiom is violated by the representation. The need to allow for this behavior in the case of monotone decreasing, but not for monotone increasing fans, explains why the representation theorems feature different continuity axioms.

\subsection{Main Results}\label{sec:results}

We present two representation theorems, one for each kind of fan. The main substantive difference in the two axiomatizations is the use of either the \owhps or \owhms axioms.

\subsubsection{Monotone increasing fan}

\begin{theorem}\label{thm:monincfan}
    A preference $\succeq$ satisfies monotonicity, strong continuity, convexity, mixing invariance, and \owhms if and only if there exists a continuous and monotone increasing fan $\Pi$ such that $\vx\succeq \vy$ if and only if $u_{\Pi}(\vx)\geq u_{\Pi}(\vy)$.
\end{theorem}

Monotone increasing fans correspond to our discussion of triage protocols and crisis standards of care. \autoref{thm:monincfan} provides an axiomatic foundation for such a social welfare function.

\subsubsection{Monotone decreasing fan}

\begin{theorem}\label{thm:welldef} If $\Pi$ is a monotone decreasing fan, then $u_{\Pi}$ is well defined \textup{(}$u_{\Pi}(\vx)\in \Re$ for all $\vx\in \Re^n_+$\textup{)}, and the preference relation defined by $u_{\Pi}$ satisfies monotonicity, convexity, \owhps, and mixing invariance. Moreover, if $\Pi(0)=\Delta$ then $u_{\Pi}$ satisfies the Inada axiom.
\end{theorem}

The main purpose of stating \autoref{thm:welldef} lies in clarifying that, despite the need for a fixed-point argument, the representation is well defined even without a continuity property.

The next result is our representation theorem for decreasing fans. It does require a continuity property, but a substantially weaker one. It is important to allow for some discontinuities, as illustrated in \autoref{fig:buminexc}.

\begin{theorem}\label{thm:main}
A preference $\succeq$ satisfies monotonicity, Inada, weak continuity, convexity, \owhps, and mixing invariance if and only if there exists an upper hemicontinuous and monotone decreasing fan $\Pi$ satisfying $\Pi(0)=\Delta$ such that $\vx\succeq \vy$ if and only if $u_{\Pi}(\vx)\geq u_{\Pi}(\vy)$.
\end{theorem}

\subsection{Remarks}\label{sec:remarks}

\begin{quoting}[leftmargin=.5cm]
\mbox{}\vspace{-\baselineskip}
\begin{description}[parsep=0pt, itemsep=2pt, leftmargin=0em, labelindent=0pt, itemindent=0pt]
\item[Inada condition] The Inada condition of \autoref{thm:main} arises because we work with the non-negative orthant rather than $\Re^n_{++}$. It is likely possible to obtain the same result without imposing an Inada condition if we work with a more restrictive domain. 
\item[Uniqueness of representation] If $u_\Pi$ and $u_{\Pi'}$ are two representations of the same preference order $\succeq$, then $u_\Pi=u_{\Pi'}$. To see why, take any $\vx\in\mathbb{R}_+^n$, and notice that we can always find a constant $c\in\mathbb{R}_+$ such that $\vx\sim \vc$. Since both  $u_\Pi$ and $u_{\Pi'}$ represent $\succeq$, we must have $u_\Pi(\vx)=u_\Pi(\vc)$ and $u_{\Pi'}(\vx)=u_{\Pi'}(\vc)$. However, by definition, $u_{\Pi}(\vc)=u_{\Pi'}(\vc)=c$, and thus, $u_{\Pi}(\vx)=u_{\Pi'}(\vx)$. Since $\vx$ was chosen arbitrarily, the two social welfare functions must coincide everywhere.

\item[Alternative notions of context] One might instead let the relevant welfare weights depend on a context other than aggregate welfare. There are several natural approaches. First, one could model the context as an exogenous state and assign a different social preference to each state. However, such a model would require additional structure justifying both the origin of the relevant states and the normative justification for the social preference associated with each one. Second, one could let the entire utility profile determine the context. This is too granular for our purposes: it permits as many contexts as utility profiles and therefore places little discipline on how the aggregation criterion varies across profiles. Third, one could index the aggregation criterion by a statistic of the utility profile. This reduces the dimensionality of the context, but the choice of statistic itself requires normative justification and may discard information central to the welfare problem. For example, average utility assigns the same context to $\vc$ and $\vx=(nc,0,\ldots,0)$ for any $c>0$, despite the sharply different distributions of utility across the two profiles. 

In contrast, our self-referential approach is self-contained: the context and its role in determining the aggregation criterion are jointly disciplined by the same axioms. Thus, all the ethical principles are explicit, transparent, and fully captured by the axioms themselves.
\end{description}
\end{quoting}

\section{Application: Ventilator allocation problem}\label{sec:application}

In this section, we provide a detailed treatment of the ventilator allocation application from \Cref{sec:motivation}. We formalize and prove the result that the optimal allocation policy is either efficient or prioritizes the worst-off, and that the comparison between the two hinges on whether the fraction of vulnerable agents surpasses a threshold. Our results justify and generalize the discussion in \Cref{sec:motivation} and the message in \autoref{fig:ventilators}.

Let $\mathcal L\coloneqq \{(a,b):a\in(0,1) \text{ and } 0<b\leq a \}$ denote the ``lower triangle'' of $(0,1)^2$. Given $(a,b)\in\mathcal L$, let $V(a,b)$ denote the unique fixed point of the mapping $v\mapsto F(v; a,b)$:
\[
V(a,b)=F\Big(V(a,b); a, b\Big).
\] The function $F$ was defined in \Cref{sec:motivation}. We have already established the existence and uniqueness of $V(a,b)$ for each $(a,b)\in \mathcal L$. For any policy $\vx$, its welfare level is $u_\Pi(\vx)=V(\bar\vx, \vx_{\min})$. Hence, two policies yielding the same mean and worst-off utilities yield the same aggregate welfare level. 

\begin{lemma}\label{lemma:monotonicity}
The fixed-point $V:\mathcal L\to (0,1)$ is continuous and strictly increasing.
\end{lemma}

In \Cref{sec:motivation}, we noted that it is without loss of generality to focus on the efficient policy $\vx^E$ (optimal  under utilitarianism) and the fair policy $\vx^F$ (optimal under Rawlsianism). The following lemma formalizes this claim by showing that any policy is welfare-dominated by either the efficient or fair policy.

\begin{lemma}\label{lemma:wlog}
For any policy $\vx$, $u_\Pi(\vx)\leq \max\{u_\Pi(\vx^E), u_\Pi(\vx^F)\}$. 
\end{lemma}

The ventilator allocation problem thus simplifies to choosing between the efficient and fair policies. The main result of this section shows that there exists a threshold determining whether $\vx^E$ or $\vx^F$ is optimal.

\begin{proposition}\label{prop:optimal threshold}
There exists a continuous and strictly increasing function $k\mapsto \alpha^*(k)$ with $\alpha^*(k)\leq k$ such that, given ventilator supply $k\in (0,1)$ and fraction $\alpha\in (0,1)$ of vulnerable agents, the welfare-maximizing policy is $\vx^F$ if $\alpha<\alpha^*(k)$ and $\vx^E$ if $\alpha >\alpha^*(k)$.
\end{proposition}

\section{Additional applications} 
While the application to healthcare triage illustrates a monotone increasing fan, other contexts call for the reverse. We discuss instances where a monotone decreasing fan is applicable: specifically, where a decrease in overall well-being triggers policy changes that favor the most vulnerable populations. Finally, in \Cref{sec:katrina} we briefly mention some additional episodes of ethical shifts in triage decisions that are captured by a monotone increasing fan.

\subsection{Educational spending}\label{sec:education}
Economic downturns typically threaten educational equity by shrinking the tax bases that fund public schools. However, policy responses are not uniform. In several documented cases, governments and districts have explicitly shifted priorities to insulate low-income and special-needs students from such shocks. Here we review some well-documented instances where priorities have shifted, and we connect them to the theory developed in the paper. In particular, it is possible to think of these shifts through the lens of the social welfare functions that we have proposed. 

First, consider the Indonesian \textit{Jaring Pengaman Sosial} (Social Safety Net) launched during the 1997 Asian Financial Crisis. With GDP contracting by over 13\%, the government faced a potential exodus of children from schools to the labor market. In response, the program allocated scholarships directly to the poorest students and sent block grants to schools in the poorest districts. As noted by \cite{Cameron2002}, the program explicitly prioritized retaining students from the lowest income quintiles. The approach was largely successful: \cite{Frankenberg1999} use longitudinal data to show that the safety net effectively smoothed consumption for poor families, rendering the decline in enrollment statistically insignificant despite the macroeconomic shock.

Second, the Latin American response to economic instability has often involved expanding Conditional Cash Transfer (CCT) programs. Programs like Brazil's \textit{Bolsa Família} and Mexico's \textit{Progresa} (later \textit{Oportunidades}) link CCT payments to school attendance. Rather than cutting these programs during economic downturns, governments have often expanded them \citep{blomquist2002social,Rawlings2005}. Research indicates that during the crises of the late 1990s and early 2000s, these programs acted as automatic stabilizers by maintaining the purchasing power of the poorest families, which ensured that children were not pulled out of school to work, a common coping mechanism during recessions in the region \citep{DeJanvry2006}.

Third, when US school districts faced both operational shutdowns and looming budget shortfalls during the COVID-19 pandemic, the American Rescue Plan (ARP) of 2021 included strict Maintenance of Equity provisions for the Elementary and Secondary School Emergency Relief (ESSER) fund. States were legally prohibited from cutting per-pupil spending in high-poverty schools at a rate faster than in the rest of the district \citep{USDeptEd2021}. Furthermore, the statute required districts to reserve at least 20\% of funds to address learning loss via interventions explicitly targeting underrepresented subgroups, including low-income families, children with disabilities, and English learners \citep{Jordan2021}. This forced a reprioritization toward the most vulnerable.\footnote{This episode is in contrast with the responses during the 2008 recession in the US, where budget cuts involved deprioritizing the worst-off students.}

These examples reveal a planner that seeks to mitigate the suffering of the worst-off agents in times of crisis. Such a preference is consistent with \owhps, which corresponds to a monotonic decreasing fan: as aggregate welfare $v$ falls, the set of welfare weights $\Pi(v)$ expands toward the edges of the simplex, forcing the planner to resolve trade-offs closer to the Rawlsian ideal precisely when resources are most scarce.

\subsection{Income distribution}\label{sec:incomeineq}

Measures of income inequality are often rooted in social welfare functions. The view that income inequality should be assessed through a normative framework dates back to \cite{daltonEJ1920}, who argued that inequality measures must be derived from an explicit aggregation procedure.\footnote{This contrasts with the ``Italian school'' (notably Corrado Gini), which favored statistical measures of dispersion without explicit welfare foundations. See \cite{atkinson2015} for an overview.} \cite{atkinson1970measurement} formalized this link, demonstrating that any inequality measure that is monotone in the Lorenz dominance order is implicitly a utilitarian aggregation of a concave utility function. Specifically, Atkinson's influential index is derived from additively aggregating utilities with constant relative risk aversion.

To formalize the application to income inequality, let $\vx \in \R^n$ denote an income distribution where $x_i$ is the income of agent $i$. We evaluate outcomes using a strictly concave function $\phi$, where the concavity reflects the degree of inequality aversion. The Atkinson approach identifies the \textit{Equally Distributed Equivalent} (EDE) income $y$, which is the income level that leaves a utilitarian planner indifferent between all agents earning $y$ and the income distribution $\vx$. The EDE is defined implicitly by:
\[
\phi(y) = \frac{1}{n}\sum_{i} \phi(\vx_i) \implies y = \phi^{-1}\left(\frac{1}{n}\sum_{i}\phi(\vx_i)\right).
\]
The Atkinson index corresponds to an income distribution's EDE for the special case when $\phi(z)=z^{1-\epsilon}/(1-\epsilon)$, where $\epsilon \geq 0$ (with $\epsilon\neq 1$) is the parameter of inequality aversion \citep{atkinson1970measurement}. When $\epsilon=0$, each agent's income  is weighted equally; as $\epsilon \to \infty$, increasingly more and more weight is placed on the poorest agent. The EDE is a form of certainty equivalent, and it also features in our formalization of the Fan social welfare function (see the definition of the representing social welfare function). 

Notice that Atkinson's measure, like most inequality indices used in practice, is homothetic: the relative inequality between $\vx$ and $\vy$ must be the same as that between $\lambda\vx$ and $\lambda\vy$ for any $\lambda>0$. Homotheticity is, however, questionable in this setting. The problem is clearly articulated by \cite{KOLM1976416}, who writes, ``In May 1968 in France, radical students triggered a student upheaval which induced a workers' general strike. All this was ended by the Grenelle agreements, which decreed a 13\% increase in all payrolls. Thus, laborers earning 80 pounds a month received $10$ pounds more, whereas executives who already earned $800$ pounds a month received $100$ pounds more. The Radicals felt bitter and cheated; in their view, this widely increased income inequality.'' \cite{sen1973economic} follows up on these ideas and argues that the absolute value of income should be taken into account in evaluating any transfer that is meant to mitigate income inequality.

The idea at the heart of these arguments is that concerns about inequality may depend on overall income levels. The recent study by \cite{costafont2025specific} documents that individuals' aversion to inequality decreases with income. Society may, therefore, arguably, have a greater concern for inequality when agents are overall poorer than when they are better off. An income distribution in which each of two agents earns \$50,000 may be preferable to a distribution in which one earns \$95,000 while the other earns \$25,000: the reduction of aggregate income by \$20,000 may be an acceptable sacrifice for the guarantee that neither agent falls too close to the poverty line. However, the same sacrifice may be unpalatable when neither agent is poor, so we may prefer a distribution in which one earns \$190,000 while the other earns \$50,000 to a distribution in which both earn \$100,000.

The Fan social welfare function accommodates these concerns by relaxing homotheticity. By imposing only \owhps, we generate a monotone decreasing fan. This implies that inequality aversion increases as society becomes poorer. Consequently, a planner may tolerate inequality in a rich society (behaving closer to a utilitarian) but demand strictly egalitarian outcomes in a poor society (approaching the Rawlsian ideal). Importantly, unlike other proposals to weaken homotheticity (see, for example, \cite{KOLM1976416}, \cite{KOLM197682}, and \cite{blackorby1980}), our approach leverages the self-referential property of the Fan social welfare function. Hence, the pivot from greater to less concern for inequality is endogenously determined by the level of societal wealth itself.

\subsection{More on triage}\label{sec:katrina}

While our discussion of triage has emphasized ventilator allocation during the COVID-19 pandemic, this is not the only relevant application of Crisis Standards of Care. A stark example of the switch from Rawlsianism to utilitarianism occurred at New Orleans' Memorial Medical Center during Hurricane Katrina in 2005.\footnote{For an account of the Memorial Medical Center episode during Katrina, see the Pulitzer Prize-winning story in ProPublica (\url{https://www.propublica.org/article/the-deadly-choices-at-memorial-826}), published in the New York Times Magazine on August 25, 2009.}

Isolated by floodwaters and facing total power failure, hospital staff were forced to improvise triage protocols. The resulting evacuation procedure explicitly deprioritized the worst-off. Patients were categorized into 1s (ambulatory), 2s (needing assistance), or 3s (critically ill). In a reversal of normal standards, 1s were evacuated first to maximize the number of survivors. Conversely, 3s, many with Do Not Resuscitate orders, were evacuated last. The hospital staff reasoned that patients categorized as 3s had the ``least to lose'' based on long-term prognosis, a utilitarian criterion typically impermissible in normal clinical settings. These ad-hoc decisions have since been formalized in Crisis Standards of Care guidelines, which explicitly acknowledge that a drop in aggregate capabilities necessitates a shift from patient-centered equity toward population-level efficacy.

This episode vividly illustrates that ethical trade-offs depend on population-wide conditions, which our model formalizes through aggregate welfare, reflecting the broader context in which these decisions are made. As aggregate welfare collapses, the decision criterion pivots toward efficiency to do the most good, even at the expense of the neediest. Our \owhms\ axiom and the resulting class of increasing Fan social welfare functions precisely capture such welfare-dependent trade-offs between utilitarianism and Rawlsianism, providing a formal rationalization for Crisis Standards of Care.

\section{Proofs}\label{sec:proofs}

The proofs of \autoref{thm:monincfan} and \autoref{thm:main} have several points in common. We first prove \autoref{thm:welldef} and \autoref{thm:main}. Then we turn to \autoref{thm:monincfan}. This section concludes with the proofs for the results in \Cref{sec:application}.

\subsection{Proof of \autoref{thm:welldef}}
Suppose that $\Pi$ satisfies the hypotheses of the theorem. For $v\in\R$, let
\[
m_v(\vx) \coloneqq \min\{\vpi\cdot \vx \mid \vpi\in \Pi(v) \},
\]
and let
\[
u(\vx)\coloneqq \inf \{v  \mid m_v(\vx) \leq v \},
\]
where we have suppressed the dependence of $u$ on $\Pi$ to simplify notation. We shall first show that $u(\vx) = \inf \{v  \mid m_v(\vx) = v \}$.

\begin{lemma}\label{lemma:lemma3} For any $\vx\in\Re^n_+$, $u(\vx)\in\Re$ is well defined, and $m_{u(\vx)}(\vx) = u(\vx)$.
\end{lemma}
\begin{proof}This proof follows roughly the proof of Tarski's fixed point theorem. 
Note first that $m_0(\vx)\geq 0$ and that for any $v\geq 0$, $m_v(\vx)\leq \max\{\vpi\cdot \vx:\vpi\in\Delta \}$. Hence for $v$ large enough ($v> \max\{\vpi\cdot \vx:\vpi\in\Delta \}$) $m_v(\vx)< v$. So the set $F_\vx \coloneqq \{v \mid m_v(\vx) \leq v \}$ is nonempty. 

Let $v^*\coloneqq \inf F_\vx$ and consider an arbitrary $v\in F_\vx$. Then $v^*\leq v$ implies that $\Pi(v^*)\supseteq \Pi(v)$, and therefore that $m_{v^*}(\vx)\leq m_{v}(\vx)\leq v$. Thus $m_{v^*}(\vx)$ is a lower bound on $F_\vx$ and we conclude that $m_{v^*}(\vx)\leq v^*$. Hence,  $v^*\in F_\vx$.

Also, for any $v<v^*$ we have $v<m_v(\vx)\leq m_{v^*}(\vx)$ as $v\notin F_\vx$ and $\Pi(v)\supseteq \Pi(v^*)$. Since this is true for arbitrary $v<v^*$ we conclude that $v^*\leq m_{v^*}(\vx)$. So $u(\vx)=v^*=m_{v^*}(\vx)$. 
\end{proof}

Clearly,  $u(\vx)\leq \inf \{v  \mid m_v(\vx) = v \}$. However, we established in \autoref{lemma:lemma3} that $u(\vx)\in \{v  \mid m_v(\vx) = v \}$. We therefore obtain $u(\vx)=\inf \{v  \mid m_v(\vx) = v \}$ as in our definition of the Fan social welfare function.
\medskip

Now we proceed to prove the rest of the statement. First, we claim that $\vx\leq \vy$ implies that $u(\vx)\leq u(\vy)$. For any $v\in F_\vy$, the monotonicity of $m_v$ implies that $m_v(\vx) \leq m_v(\vy) \leq v$. So $F_\vy\subseteq F_\vx$ and hence $u(\vx)\leq u(\vy)$.

Second, we show that $u(\vc)=c$ for any $c\in\R_+$. This property is useful in the sequel. For any $v$, $m_v(\vc)=c$. So $F_{\vc}=\{v \mid c\leq v \}$, and therefore $u(\vc)=\inf F_{\vc} = c$.

Third, $u(\vx)\geq u(\vc)$ and $u(\vy)\geq u(\vc)$ imply that $u(\al \vx + (1-\al) \vy)\geq u(\vc)$ for any $\al\in (0,1)$. To see this, fix  $v<c$.  Then $v<m_v(\vx)$ and $v<m_v(\vy)$, as $u(\vx)\geq u(\vc)=c$ and $u(\vy)\geq u(\vc)=c$. Let $\al\in (0,1)$. By the concavity of $m_v$ (which is the infimum of a collection of affine functions, and hence concave),
\[
m_v(\al \vx + (1-\al) \vy)\geq \al m_v(\vx) + (1-\al) m_v(\vy)>v.
\] Hence $v\notin F_{\al \vx + (1-\al) \vy}$ for any $v<c$ and thus $u(\al \vx + (1-\al) \vy)\geq c$. 

Fourth, we show that $u(\vx)> u(\vc)$ and $\la>1$ implies $u(\la \vx)> u(\la \vc)$. Fix $v\leq c$ and $\la>1$. Then $m_v(\vx)> v$ as $u(\vx)> u(\vc)=c$. 
Now, since $\la>1$ we have $\Pi(v)\supseteq \Pi(\la v)$, and thus $m_v(\la \vx)\leq m_{\la v}(\la \vx)$. Then
\[
m_{\la v}(\la \vx)\geq  m_v(\la \vx) = \la m_v(\vx)>\la v,
\] where we have used the homogeneity of $m_v$.
Thus $\la v\notin F_{\la \vx}$ for any $v\leq c$. Equivalently, for any $v'\leq \la c$ we can write $v'=\la v$ with $v\leq c$ and conclude that $v'\notin F_{\la \vx}$. Thus $u(\la \vx)> \la c$.

Finally, we claim that $u(\vx) =  u(\vc)$ and $\al\in (0,1)$ imply $u(\al \vx + (1-\al) \vc) = u(\vc)$. To prove this, recall that we have already established convexity, so $u(\al \vx + (1-\al) c) \geq u(\vc)$. Suppose, towards a contradiction, that $u(\al \vx + (1-\al) c) > u(\vc)$. Then $c=u(\vc)$ means that $c\notin F_{\al \vx + (1-\al) c}$. 
Observe that $m_v(\al \vx + (1-\al) c) = \al m_v(\vx) + (1-\al) c$ for any $v$ because if $\vpi\in \Pi(v)$ then $\vpi\cdot \al \vx + (1-\al) c = \al \vpi\cdot \vx + (1-\al) c$. 
Then
\[
c < m_c(\al \vx + (1-\al) c) = \al m_c(\vx) + (1-\al) c.
\]
So $c< m_c(\vx)$, which contradicts the fact that $u(\vx)=u(\vc)=c$. 
Hence, $u(\vx) =  u(\vc)$ implies that $u(\al \vx + (1-\al) \vc) = u(\vc)$.

To prove the last statement in the theorem, note that when $\Pi(0)=\Delta$, for any $\vx\in\Re^n_+\setminus\Re^n_{++}$ we have $m_0(\vx)=0$. Hence $u(\vx) = \inf \{v  \mid m_v(\vx) \leq v \}=0$.

\subsection{Proof of \autoref{thm:main}}

\subsubsection{Sufficiency}
We first show that a preference $\succeq$ that satisfies monotonicity, Inada, weak continuity, convexity, mixing invariance, and \owhps has a utility representation. To that end, for each $\vx\in\R^n_+$, the set $\{c\in\Re:\vc\succeq \vx \}$ is nonempty by monotonicity. Let $e_\vx$ be the infimum of this set. The scalar $e_{\vx}$ is an \df{egalitarian certainty equivalent} of $\vx$. By the weak continuity axiom, $\ve_\vx\succeq \vx$. If $\ve_\vx\succ \vx$, then again by weak continuity there exists $c'<e_\vx$ with $\vc'\succ \vx$, which contradicts the fact that $e_\vx$ is the infimum. Hence, $\ve_\vx\sim \vx$. By monotonicity, $e_\vx$ is unique. Note that $\vx\mapsto e_\vx$ is a utility representation of $\succeq$ because of monotonicity.
\smallskip

Next, we show that the axioms imply the existence of a monotone decreasing and upper hemicontinuous fan $\Pi$. For each $c$, define $U^\circ(c)\coloneqq\{\vx\in\Re^n_+ \mid \vx\succ \vc \}$. Note that $U^\circ(c)$ is convex and nonempty. Furthermore, $U^\circ(c)$ is open by the continuity axiom. Thus, there is a hyperplane properly separating $U^\circ(c)$ and $\{\vc\}$, i.e., there is a non-zero vector $\vpi\in\R^n$ such that $\vpi\cdot U^\circ(c)\geq \vpi\cdot \vc$ and $\vpi\cdot \vx> \vpi\cdot \vc$ for some $\vx\in U^\circ(c)$. By monotonicity of $\succeq$, we obtain that $\vpi>0$. We may then normalize $\vpi$ so that $\vpi\in\Delta$, which implies $\vpi\cdot \vc=c$. 

For each $c\geq 0$, define $\Pi(c)\coloneqq\{\vpi\in\Delta:\vpi\cdot \vx\geq c \text{ for all } \vx\in U^\circ(c)\}$. Observe that $\Pi(c)$ is nonempty, closed and convex. We proceed to prove the two properties required of $c\mapsto\Pi(c)$:
\smallskip

\noindent\textbf{Monotone decreasing:} Consider the case where $c=0$. By the Inada and monotonicity axioms, $U^\circ(0)=\R^n_{++}$. For any $\vx\in\R^n_{++}$ and any $\vpi\in\Delta$, we have $\vpi\cdot\vx\geq 0$. Hence, $\Pi(0)=\Delta$, and $\Pi(c')\subseteq \Pi(0)$ for all $c'>0$ trivially. 

Next, consider the case where $c>0$. Take any $c'>c$ and choose any $\vpi\in\Pi(c')$. If $\vx\in U^\circ(c)$ then $(c'/c)\vx\succ (c'/c) \vc=\vc'$ by \owhps. Thus,
$(c'/c)\vx\in U^\circ(c')$ and $\vpi\cdot [(c'/c)\vx]\geq c'$. Hence $\vpi \cdot \vx\geq c$. Since $\vx\in U^\circ(c)$ was arbitrary, we have $\vpi\in\Pi(c)$.
\smallskip

\noindent\textbf{Upper hemicontinuity:}  Let $c^k\to c$ and choose a convergent sequence $\vpi^k\in\Pi(c^k)$, with $\vpi=\lim \vpi^k$. For any $\vx\succ \vc$, we have $e_\vx>c$. For large enough $k$, we then have  $e_\vx> c^k$, implying $\vx\succ \vc^k$. Thus $\vpi^k\cdot \vx\geq c^k$ for $k$ large enough. Consequently, $\vpi\cdot \vx\geq c$. Since $\vx\in U^\circ(c)$ was arbitrary, we conclude that $\vpi\in\Pi(c)$. 
\smallskip

We have thus far established the existence of a monotone decreasing and upper hemicontinuous $\Pi$. We now prove some properties of $\Pi$ that are useful to establish the desired utility representation. 

\Claim For any $\vx\succ \vc$, we have that $\vpi\cdot \vx> c$ for all $\vpi\in \Pi(c)$. To see why, note that $U^\circ(c)$ is open by the weak continuity axiom. Thus, we can find $\vx'\in U^\circ(c)$ with $\vx'\ll \vx$. Consequently, for all $\vpi\in \Pi(c)$, we have $\vpi\cdot \vx>\vpi\cdot \vx'\geq c$.

\smallskip

\Claim For any $\vx\sim \vc$, we have that $\vpi\cdot\vx\geq c$ for all $\vpi\in\Pi(c)$. Moreover, there exists a $\vpi\in\Pi(c)$ such that $\vpi\cdot \vx=c$.

To show the first part, suppose (for the sake of contradiction) that $\vpi\cdot\vx<c$ for some $\vpi\in\Pi(c)$. This is only possible if $c>0$. We can then choose $\vx'\gg\vx$ with $\vx'\in\R^n_+$ and $\vpi\cdot \vx'<c$. However, by the monotonicity axiom, $\vx'\succ\vx\sim \vc$, which implies that $\vx'\in U^\circ(c)$ and $\vpi\cdot\vx'\geq c$. A contradiction.

To show the second part, define $L\coloneqq\{\al \vx + (1-\al) \vc:\al\in [0,1] \}$. For any $\vz\in L$, we have $\vz\sim \vc$ by the mixing invariance axiom. Hence, $\vz\notin U^\circ(c)$, making $L$ and $U^\circ(c)$ disjoint convex sets. Again by the separating hyperplane theorem, there exists $\vpi_\vx\neq 0$ with $\vpi_\vx\cdot L\leq \vpi_\vx \cdot U^\circ(c)$. By the monotonicity of $\succeq$ we may also take $\vpi_\vx>0$. Hence, after a normalization, $\vpi_\vx\in\Delta$.

Note that if $\vpi_\vx\cdot \vx> c$ then we can find $c'>c$ with  $\vpi_\vx\cdot \vx> c'=\vpi_\vx\cdot \vc'$, which would be a contradiction as $\vc'\in U^\circ(c)$ while $\vx\in L$.  Note also that $\vpi_\vx\cdot \vx< c$ leads to a contradiction because we may choose $\vx'\gg \vx$ with $\vpi_\vx\cdot \vx'< c$; again a contradiction as monotonicity implies $\vx'\succ \vx$ and thus by transitivity $\vx'\in U^\circ(c)$ while $\vc\in L$. We thus conclude that $\vpi_\vx\cdot \vx=c$. Then for all $\vz\in U^\circ(c)$,  $\vpi_\vx\cdot \vz\geq \vpi_\vx\cdot \vx=c$, which by definition of $\Pi(c)$ implies that $\vpi_\vx\in\Pi(c)$. 
\smallskip


We now establish that $\succeq$ has a utility representation of the form $u_\Pi$. Let $u(\vx)=e_\vx$. We claim that $u(\vx) = \min\{v  \mid m_v(\vx)\leq v\}$. Since $\vx\sim e_\vx$,  we have $\vpi\cdot \vx\geq \vpi\cdot \ve_\vx=e_\vx$ for all $\vpi\in \Pi(e_\vx)$, with equality for some $\vpi\in\Pi(e_\vx)$ as identified in the previous argument. Thus $e_\vx = m_{e_\vx}(\vx)$.

Let $v$ be such that $m_v(\vx) \leq v$. 
Suppose, towards a contradiction, that $v<e_\vx$. Then $\Pi(v)\supseteq \Pi(e_\vx)$ implies $m_v(\vx)\leq m_{e_\vx}(\vx)$. 
But $v<e_\vx$ means $\vx\in U^\circ(v)$ and hence $\vpi\cdot \vx>v$ for all $\vpi\in\Pi(v)$. Since $\Pi(v)$ is compact, $m_v(\vx)> v$; a contradiction. This implies that $u(\vx) = \min\{v  \mid m_v(\vx)\leq v\}$, and therefore that $u=u_{\Pi}$.


\subsubsection{Necessity}

We established the necessity of most of the axioms in \autoref{thm:welldef}. Only continuity is missing.

\begin{lemma}\label{lem:infm} If $u(\vx)>c$, then $\inf\{m_v(\vx) - v :v\in[0, c] \}>0$.
  \end{lemma}

\begin{proof}
Observe that  $u(\vx)>c$ implies that for any $v\leq c$, $m_v(\vx)>v$.  To show that $\inf\{m_v(\vx) - v :v\in[0, c] \}>0$, consider a sequence $v_n$ in $[0,c]$. By compactness of $[0,c]$, we may choose $v_n$ to be convergent and let $v=\lim v_n$. Notice that we have $m_{v_n}(\vx)-v_n>0$ for all $n$, and $m_v(\vx)-v>0$.

Since $\Pi(v_n)$ is closed for each $n$, there exist $\vpi_n\in\Pi(v_n)$ such that $m_{v_n}(\vx) = \vpi_n \cdot \vx$. By compactness of $\Delta$, we may assume $\vpi_n\to\vpi$. And by upper hemicontinuity of $\Pi$, we know that $\vpi\in\Pi(v)$. Hence, $\vpi\cdot \vx\geq m_v(\vx)$. Consequently, 
\[
\lim m_{v_n}(\vx)-v_n=\vpi\cdot \vx-v\geq m_v(\vx)-v>0.
\]
\end{proof}

\begin{lemma}\label{lem:cont} If $u(\vx)>u(\vc)$, then there exists $\ep>0$ so that if $\norm{\vx-\vy}<\ep$ then $u(\vy)>u(\vc)$.
  \end{lemma}

\begin{proof}
  From \autoref{lem:infm}, we may choose $\ep$ with
  \[
0<\ep < \inf\{m_v(\vx) - v :v\in[0, c] \}.
  \]
  Fix $\vy\in\Re^n_+$ with  $\norm{\vx-\vy}<\ep$. Choose $\vpi_\vx,\vpi_\vy\in\Pi(v)$ with $m_v(\vx)=\vpi_\vx\cdot \vx$ and $m_v(\vy)=\vpi_\vy\cdot \vy$. Then  
  \begin{align*}
    \abs{m_v(\vx) - m_v(\vy)} & = \abs{\vpi_\vx \cdot \vx - \vpi_\vy \cdot \vy}\\
   & = \max\{\vpi_\vx \cdot \vx - \vpi_\vy \cdot \vy, \,  \vpi_\vy \cdot \vy - \vpi_\vx \cdot \vx\} \\
    & \leq \max\{\vpi_\vy \cdot( \vx - \vy), \,  \vpi_\vx \cdot (\vy - \vx)\} \\
    & \leq \max\left\{ \max_{\vpi\in\Pi(v)}\vpi\cdot (\vx - \vy), \,   \max_{\vpi\in\Pi(v)}\vpi \cdot (\vy - \vx) \right\}\\
    & \leq \norm{\vx-\vy}_{\infty}\\
    &\leq \norm{\vx-\vy}\\
    &<\ep,
  \end{align*}
where the first inequality follows because $\vpi_\vx\cdot \vx\leq \vpi_\vy\cdot \vx$ and $\vpi_\vy\cdot \vy\leq \vpi_\vx\cdot \vy$, and the third inequality follows because
\[ 
\max_{\vpi\in\Pi(v)}\vpi\cdot (\vx - \vy)\leq \max_{\vpi\in\Delta}\vpi\cdot (\vx - \vy)\leq \norm{\vx-\vy}_{\infty},\]
and similarly for $\max\{\vpi \cdot (\vy - \vx):\vpi\in\Pi(v)\}$.

Then for any $v\in [0,c]$, $m_v(\vy) - v > m_v(\vx) -\ep - v > 0$ by definition of $\ep$. This means that $u(\vy)> c=u(\vc)$.
\end{proof}

The first part of the weak continuity axiom is exactly \autoref{lem:cont}. For the second part, note that $u(\vc)\geq u(\vx)$ if and only if $c\geq u(\vx)$; and that $u(\vc)\leq u(\vx)$ if and only if $c\leq u(\vx)$. 

\subsection{Proof of \autoref{thm:monincfan}}

\subsubsection{Sufficiency}
We first show that a preference $\succeq$ that satisfies monotonicity,  strong continuity, convexity, mixing invariance, and \owhms has a utility representation. For each $\vx\in\R^n_+$, let $e_\vx = \inf \{c\in\R:\vc\succeq \vx \}$. Then $\vx\mapsto e_\vx$ is a utility representation of $\succeq$. The argument is the same as for \autoref{thm:main} and therefore omitted.
\smallskip

Next, we show that the axioms imply the existence of a monotone increasing and continuous fan $\Pi$. For each $c$, define $U^\circ(c)\coloneqq\{\vx\in\Re^n_+ \mid \vx\succ \vc \}$. Note that $U^\circ(c)$ is convex and nonempty. Furthermore, $U^\circ(c)$ is open by the continuity axiom. Thus, there is a hyperplane properly separating $U^\circ(c)$ and $\{\vc\}$, i.e., there is a non-zero vector $\vpi\in\R^n$ such that $\vpi\cdot U^\circ(c)\geq \vpi\cdot \vc$ and $\vpi\cdot \vx> \vpi\cdot \vc$ for some $\vx\in U^\circ(c)$. By monotonicity of $\succeq$, we obtain that $\vpi>0$. We may then normalize $\vpi$ so that $\vpi\in\Delta$, which implies $\vpi\cdot \vc=c$. 

For $c> 0$, define $\Pi(c)\coloneqq\{\vpi\in\Delta:\vpi\cdot \vx\geq c \text{ for all } \vx\in U^\circ(c)\}$. Observe that $\Pi(c)$ is nonempty, closed, and convex. We proceed to prove the two properties required of $c\mapsto\Pi(c)$:

\smallskip

\noindent\textbf{Monotone increasing:} Consider the case where $c>0$. Take any $c'>c$ and choose any $\vpi\in\Pi(c)$. If $\vx\in U^\circ(c')$ then $(c/c')\vx\succ (c/c') \vc'=\vc$ by \owhms. Thus,
$(c/c')\vx\in U^\circ(c)$ and $\vpi\cdot [(c/c')\vx]\geq c$. Hence $\vpi \cdot \vx\geq c'$. Since $\vx\in U^\circ(c')$ was arbitrary, we have $\vpi\in\Pi(c')$.

Now we may define $\Pi(0)\coloneqq\bigcap_{c>0} \Pi(c)$. By the finite intersection property, $\Pi(0)\neq \emptyset$. We also get that $\Pi(0)\subseteq\Delta$ is closed and convex. Since $U^\circ(0)$ is open, for each $\vx\in U^\circ(0)$, there exists a small enough $c>0$ such that $\vx\in U^\circ(c)$. Thus, for all $\vpi\in\Pi(0)$, we obtain $\vpi\cdot \vx>c>0$.
\smallskip

\noindent\textbf{Upper hemicontinuity:}  Let $c^k\to c$ and choose a convergent sequence $\vpi^k\in\Pi(c^k)$, with $\vpi=\lim \vpi^k$. Consider first the case where $c>0$. For any $\vx\succ \vc$, we have $e_\vx>c$. For large enough $k$, we then have  $e_\vx> c^k$, implying $\vx\succ \vc^k$. Thus $\vpi^k\cdot \vx\geq c^k$ for $k$ large enough. Consequently, $\vpi\cdot \vx\geq c$. Since $\vx\in U^\circ(c)$ was arbitrary, we conclude that $\vpi\in\Pi(c)$. 

Next consider the case where $c=0$. To show $\vpi \in \Pi(0)$, we must show $\vpi \in \Pi(v)$ for every $v > 0$. Since $c^k \to 0$, for any fixed $v > 0$, eventually $c^k < v$. Because $\Pi$ is monotone increasing, $\Pi(c^k) \subseteq \Pi(v)$. Thus $\vpi^k \in \Pi(v)$ for all large $k$. Since $\Pi(v)$ is closed, the limit $\vpi \in \Pi(v)$. This holds for all $v>0$, so $\vpi \in \bigcap_{v>0} \Pi(v) = \Pi(0)$.
\smallskip

We have thus far established the existence of a monotone increasing and upper hemicontinuous $\Pi$. We now prove some properties of $\Pi$ that are useful to establish that $\Pi$ is lower hemicontinuous.

\Claim For any $\vx\succ \vc$, we have that $\vpi\cdot \vx> c$ for all $\vpi\in \Pi(c)$. The argument is the same as \autoref{thm:main} and is thus omitted.
\smallskip

\Claim For any  $\vx\sim \vc$, we have that $\vpi\cdot\vx\geq c$ for all $\vpi\in\Pi(c)$. Moreover, there exists a $\vpi\in\Pi(c)$ such that $\vpi\cdot \vx=c$. The argument is the same as \autoref{thm:main} and is thus omitted.
\smallskip

\Claim For any $\vx\in\R^n_+$ such that $\vpi\cdot\vx\geq c$ for all $\vpi\in\Pi(c)$, we have that $\vx\succeq \vc$. Suppose (for the sake of a contradiction) that $\vc\succ \vx$. Then $\vx\sim \ve_\vx \ll \vc$. By the previous argument, there exists $\vpi\in \Pi(e_\vx)$ with $\vpi\cdot \vx = e_\vx$. But since $\Pi(e_\vx)\subseteq \Pi(c)$, we have found  $\vpi\in \Pi(c)$ such that $\vpi\cdot \vx<c$. A contradiction.
\smallskip

Putting the three claims together, we obtain that $\vx\succeq \vc$ if and only if $\vpi\cdot\vx\geq c$ for all $\vpi\in\Pi(c)$. We now prove that $\Pi$ is lower hemicontinuous, and thus, continuous.
\smallskip

\noindent\textbf{Lower hemicontinuity:} Let $c^k\to c$ and pick any $\vpi\in \Pi(c)$. Let $N_{\vpi}$ be an open ball with center $\vpi$. We claim that, infinitely often, $\Pi(c^k)\cap N_{\vpi}$ is nonempty. 

For any $c^k\geq c$, it is immediate that $\vpi\in\Pi(c)\subseteq\Pi(c^k)$ and there is nothing to prove. In particular, this trivially proves the claim for the case where $c=0$. Let us thus assume that $c^k<c$ for all $k$, which implies that $c>0$. In fact, by considering a subsequence we may assume that $c^k \uparrow c$.

Suppose, towards a contradiction, that $$\left( \bigcup_{k}\Pi(c^k) \right)\cap N_{\vpi} = \os.$$ Since $\Pi$ is monotone, proving that this leads to a contradiction implies that $\Pi(c^k)$ eventually intersects $N_{\vpi}$; so they intersect infinitely often. 

Observe that $\cup_{k}\Pi(c^k)$ is convex by the monotonicity of $\Pi$ and the convexity of each $\Pi(c^k)$. Let $C$ be the conic hull of this set; so $$C=\left\{\la \vx \mid \la\geq 0 \text{ and } \vx \in \bigcup_{k}\Pi(c^k) \right\}.$$ Similarly, let $C'$ be the conic hull of $N_{\vpi}$. The sets $C$ and $C'$ are convex cones that share only the null vector. Hence, by the separating hyperplane theorem, there exists $\vz\in\Re^n$ with $\vz\cdot N_{\vpi}\leq \vz\cdot \cup_{k}\Pi(c^k)$ and $\vz\neq 0$. Since $\{0\}=C\cap C'$,  and $N_{\vpi}$ is an open ball, we have $\vz\cdot \vpi<0\leq \vz\cdot \vpi'$ for all $\vpi'\in \cup_k\Pi(c^k)$. Finally, observe that we can choose $\vz$ with arbitrarily small sup norm (arbitrarily small $\norm{\vz}_{\infty}$) so that $\vc+\vz\gg 0$.

Note that $\vpi \cdot \vz<0$ and $\vpi\in \Pi(c)$ imply that $\vc\succ \vc+ \vz$. Similarly, $\vpi'\cdot \vz\geq 0$ for all $\vpi'\in \cup_k \Pi(c^k)$ implies that $\vc^k+\vz\succeq \vc^k$ for all $k$. But the continuity axiom implies that the graph of the preference relation $\succeq$ is closed. So $c=\lim_{k\to\infty} c^k$ and  $\vc^k + \vz\succeq \vc^k$ for all $k$ contradict that $\vc\succ \vc+\vz$.  We conclude that $(\cup_{k}\Pi(c^k))\cap N_{\vpi} \neq \os$. Since the ball $N_{\vpi}$ was arbitrary, $\Pi$ is lower hemicontinuous.
\smallskip

We are now in a position to wrap up the proof of the sufficiency direction. It remains to show that $\succeq$ has a utility representation of the form $u_{\Pi}$.

For any $\vx\in\R^n_+$, the function $v\mapsto m_v(\vx)\coloneqq\inf\{\vpi\cdot \vx:\vpi\in \Pi(v) \}$ is monotone decreasing and continuous. It is monotone decreasing because $\Pi$ is monotone increasing. It is continuous by the maximum theorem, as $\Pi$ is a continuous (both lower and upper hemicontinuous) correspondence. We also have $m_0(\vx)\geq 0$ and $m_v(\vx)<v$ for $v$ large enough. By the intermediate value theorem, there is $v\geq 0$ with $v=m_v(\vx)$; a fixed point. Since the function is monotonically decreasing, the fixed point is unique.

Let $u(\vx)=e_\vx$. Since $\vx\sim e_\vx$, we have $\vpi\cdot \vx\geq e_\vx$ for all $\vpi\in \Pi(e_\vx)$ with equality for some $\vpi\in\Pi(e_\vx)$. Thus $e_\vx = m_{e_\vx}(\vx)$. Since the fixed point of $v\mapsto m_v(\vx)$ is unique, we have that 
\[
\inf\{v \mid v=m_v(\vx) \} = e_\vx = u(\vx).
\]

\subsubsection{Necessity}

It is important to remember that, in this version of the model where the fan is monotone increasing, there exists a unique fixed point $v=m_v(\vx)$ and that for any $v'<v$, $m_{v'}(\vx)>v'$.  For notational simplicity, here we write $u$ for $u_{\Pi}$. We proceed to show that the representation satisfies the axioms.
\bigskip

\noindent\textbf{Monotonicity:} \\
Let $\vy\geq \vx$. If $u(\vy)<u(\vx)$, we get the following contradiction:
\[
u(\vx)=m_{u(\vx)}(\vx)\leq m_{u(\vx)}(\vy)\leq m_{u(\vy)}(\vy)=u(\vy)<u(\vx),
\]
where the first inequality follows from the fact that $\vx\leq \vy$ and the second one follows because $\Pi(u(\vy))\subseteq \Pi(u(\vx))$ since $\Pi$ is monotone increasing. Hence, $\vy\geq \vx$ implies $u(\vy)\geq u(\vx)$.

Similarly, let $\vy\gg \vx$. If $u(\vy)\leq u(\vx)$, we get the following contradiction:
\[
u(\vx)=m_{u(\vx)}(\vx)< m_{u(\vx)}(\vy)\leq m_{u(\vy)}(\vy)=u(\vy)\leq u(\vx),
\]
where the first inequality follows from the fact that $\vx\ll \vy$ and the second one follows again because $\Pi$ is monotone increasing. Hence, $\vy\gg \vx$ implies $u(\vy)>u(\vx)$.
\bigskip

\noindent\textbf{Continuity:} \\
First we show that upper contour sets are closed. Fix an arbitrary $\vx$. 
Let $\vy^k$ be a sequence with $u(\vy^k)\geq u(\vx)$ and $\vy=\lim_{k\to\infty } \vy^k$. Suppose, towards a contradiction, that $u(\vy)<u(\vx)$. Then 
\[
m_{u(\vy)}(\vy^k) \geq m_{u(\vy^k)}(\vy^k) = u(\vy^k)\geq u(\vx),
\] as $u(\vy)< u(\vy^k)$. But the maximum theorem implies that $m_{u(\vy)}(\vy)=\lim m_{u(\vy)}(\vy^k)$. So we obtain $u(\vy)=m_{u(\vy)}(\vy)\geq u(\vx)>u(\vy)$, which is absurd.

Second, we show that lower contour sets are closed. 
Let $\vy^k$ be a sequence with $u(\vy^k)\leq u(\vx)$ and $\vy=\lim_{k\to\infty } \vy^k$. Suppose, towards a contradiction, that $u(\vy)>u(\vx)$. Note that $u(\vy^k)\geq 0$ by definition of $u$, so the sequence $(u(\vy^k))_{k=0}^\infty$ is bounded. Let $\ep>0$ be such that $u(\vy)-\ep > u(\vx)$. Then,
\[\begin{split}
    u(\vy)=m_{u(\vy)}(\vy)\leq m_{u(\vx)}(\vy)=\lim_{k\to\infty} m_{u(\vx)}(\vy^k) &\leq \liminf_{k\to\infty }m_{u(\vy^k)}(\vy^k) \\ &= \liminf_{k\to\infty } u(\vy^k),
\end{split}\]where the first inequality is because $u(\vx)<u(\vy)$ and the second because $u(\vy^k)\leq u(\vx)$. The second equality follows from the continuity of $\vy\mapsto m_{u(\vx)}(\vy)$. The lim inf is finite because the sequence is bounded. This is the desired contradiction because $u(\vy^k)<u(\vy)-\ep$ for all $k$.
\bigskip

\noindent\textbf{Convexity:} \\
Fix $\vx,\vy\in \Re^n_+$ and $\la\in (0,1)$. 
Suppose that $$u(\la \vx + (1-\la) \vy)=m_c(\la \vx + (1-\la) \vy)=c< c'.$$ Then there exists $\vpi\in\Pi(c)$ with $\vpi\cdot (\la \vx+(1-\la)\vy)=c$. So either $\vpi\cdot \vx\leq c$ or $\vpi\cdot \vy\leq c$. Suppose without loss of generality that it is the former. Then since $\vpi\in\Pi(c)\subseteq \Pi(c')$ we have that $m_{c'}(\vx) \leq c< c'$. Thus we cannot have $u(\vx)\geq c'$. Conclude that if $u(\vx)\geq c'$ and $u(\vy)\geq c'$, then $u(\la \vx + (1-\la) \vy)\geq c'$.
\bigskip

\noindent\textbf{Mixing invariance:} \\
Let $u(\vx)=c$. Then $m_c(\vx)=c$, so there exists $\vpi_\vx\in \Pi(c)$ with $\vpi_\vx\cdot \vx=c$ while $\vpi\cdot \vx\geq c$ for all $\vpi\in\Pi(c)$. Then for any $\la\in (0,1)$ we have $\vpi_\vx\cdot (\la \vx+(1-\la)\vc)=c$ while $\vpi\cdot (\la \vx+(1-\la)\vc)\geq c$ for all $\vpi\in\Pi(c)$. Thus $m_c(\la \vx+(1-\la)\vc)=c$ and therefore $u(\la \vx+(1-\la)\vc)=c$ as the fixed point of $v\mapsto m_v(\vz)$ is unique.
\bigskip

\noindent\textbf{\owhm:} \\
Suppose that $u(\vx)>c$ and $\la \in (0,1)$. Then $m_c(\vx)>c$, so
for all $\vpi\in\Pi(c)$ we have that $\vpi\cdot \vx>c$. But since $\la c\leq c$ (allowing for $c=0$), we have for $\vpi\in\Pi(\la c)\subseteq \Pi(c)$ that $\la \vpi\cdot \vx>\la c$. Thus $m_{\la \vc}(\la \vx)>\la c$ and therefore $u(\la \vx)>\la c$.

\subsection{Proof of \autoref{lemma:monotonicity}}

The mapping $(a,b)\mapsto F(\cdot;a,b)$ is continuous, and $F(\cdot;a,b)$  has a unique fixed point $V(a,b)$ for each $(a,b)\in \mathcal{L}$. 
The function $v\mapsto F(v;a,b)$ is monotone decreasing, so $V(a,b)\in [0,F(0;a,b)]\subseteq [0,1]$ for all $a,b$. Let $(a^k,b^k)\to (a,b)$. By compactness of $[0,1]$ there is a limit point $V^*$ of $V(a^k,b^k)$; by going to a subsequence, we have $V^*=\lim_{k\to \infty} V(a^k,b^k)$. Now continuity of $F$ and $V(a^k,b^k)=F(V(a^k,b^k);a^k,b^k)$ for all $k$ imply that $V^*$ is a fixed point of $v\mapsto F(v;a,b)$. Since fixed points are unique, we have $V^*=V(a,b)$, and since $V^*$ was an arbitrary limit point, we conclude that $V(a,b)=\lim_{k \to \infty} V(a^k,b^k)$. Thus the mapping $(a,b)\mapsto V(a,b)$ is continuous over $\mathcal{L}$. Moreover, for any $(a,b), (a',b')\in\mathcal{L}$ with $(a,b)>(a',b')$, we have $F(v; a,b)>F(v;a',b')$ for all $v\in (0,1)$, so $V(a,b)>V(a',b')$.

\subsection{Proof of \autoref{lemma:wlog}}

Fix any policy $\vx$. First, consider the case where not all type $L$ agents are treated under this policy. Then $\vx_{\min}=L$. Since $k<1$ and the efficient policy prioritizes treating type $H$ agents, we also have $\vx^E_{\min}=L$. By definition, the efficient policy maximizes mean utility, so $\bar\vx^E\geq \bar\vx$. By \autoref{lemma:monotonicity}, 
\[
u_\Pi(\vx)=V(\bar \vx, L)\leq V(\bar\vx^E, L)=u_\Pi(\vx^E). 
\]

Next, consider the case where all type $L$ agents are treated. This is only possible if $k\geq \alpha$, i.e., there are enough ventilators to treat all vulnerable patients. Then $\vx_{\min}=\min\{\gamma L, H\}$. Since $k\geq \alpha$, all type $L$ agents are also treated under the fair policy, so $\vx^F_{\min}=\min\{\gamma L, H\}$. By definition, the fair policy maximizes mean utility subject to treating as many type $L$ agents as possible, so $\bar\vx^F\geq \bar\vx$. By \autoref{lemma:monotonicity}, 
\[
u_\Pi(\vx)=V(\bar \vx, \min\{\gamma L, H\})\leq V(\bar\vx^F, \min\{\gamma L, H\})=u_\Pi(\vx^F). 
\]

\subsection{Proof of \autoref{prop:optimal threshold}}
We divide the $k$-$\alpha$ parameter space into the ``lower triangle'' $\mathcal{L}$ of $(0,1)^2$ and its ``upper triangle'' $\mathcal{U}\coloneqq (0,1)^2\backslash \mathcal{L}$. 

First, consider the case where $(k,\alpha)\in \mathcal U$, so that not all type $L$ agents can be  treated under any policy. Then $\vx^E_{\min}=\vx^F_{\min}=L$. However, $\bar\vx^E>\bar\vx^F$. By \autoref{lemma:monotonicity}, $u_\Pi(\vx^E)> u_\Pi(\vx^F)$.
    
Next, consider the case where $(k,\alpha)\in \mathcal L$. The fair policy can treat all type $L$ agents, while some type $L$ agents remain untreated under the efficient policy. Thus, $\vx^E_{\min}=L<\min\{\gamma L, H\}=\vx^F_{\min}$. Moreover,
\[
\bar\vx^E=\underbrace{\alpha L+(1-\alpha)H+(\gamma-1)kH+(\gamma-1)(H-L)\min\{1-\alpha-k,0\}}_{\eqqcolon \mathcal E (k,\alpha)}.
\]
Intuitively, $\alpha L+(1-\alpha)H$ is the baseline welfare level without any treatment. If we could treat a mass $k$ of type $H$ agents, welfare would rise by $(\gamma-1)kH$. However, we can only do so if $1-\alpha\geq k$. Otherwise, we can only treat a mass $1-\alpha$ of type $H$ agents, and the remaining $k-(1-\alpha)$ ventilators treat type $L$ agents. The last term in $\bar\vx^E$ captures this adjustment. 

On the other hand, 
\[
\bar\vx^F= \underbrace{\alpha L+(1-\alpha)H+(\gamma-1)\alpha L+(\gamma-1)(k-\alpha)H}_{\eqqcolon \mathcal F (k,\alpha)}.
\]
Since $(k,\alpha)\in \mathcal L$, we can treat all type $L$ agents under the fair policy, raising the baseline welfare level by $(\gamma-1)\alpha L$. The remaining $k-\alpha$ ventilators treat type $H$ agents, further increasing welfare by $(\gamma-1)(k-\alpha)H$.

Both $\mathcal{E}$ and $\mathcal{F}$ are continuous over $\mathcal L$. Since $\vx^E$ maximizes mean utility while $\vx^F$ maximizes worst-off utility, we have
\[
\mathcal{E}(k,\alpha)-L>\mathcal{F}(k,\alpha)-\min\{\gamma L,H\}
\]
for all $(k,\alpha)\in\mathcal{L}$. Additionally, $\frac{\partial}{\partial \alpha} \mathcal{F}(k,\alpha)=\gamma(L-H)$ and
\[
\frac{\partial}{\partial \alpha} \mathcal{E}(k,\alpha)=\left\{\begin{array}{ccc}
    L-H &\mbox{if}& k<1-\alpha  \\
     \gamma(L-H)&\mbox{if}& k> 1-\alpha
\end{array}\right. ,
\]
so for each $k\in (0,1)$, both $\mathcal F(k,\cdot)$ and $\mathcal{E}(k,\cdot)$ are strictly decreasing over $(0,k)$. Similarly, $\frac{\partial}{\partial k} \mathcal{F}(k,\alpha)=(\gamma-1)H$ and
\[
\frac{\partial}{\partial k} \mathcal{E}(k,\alpha)=\left\{\begin{array}{ccc}
    (\gamma-1)H &\mbox{if}& k<1-\alpha  \\
    (\gamma-1)L&\mbox{if}& k> 1-\alpha
\end{array}\right. ,
\]
so for each $\alpha\in (0,1)$, both $\mathcal F(\cdot,\alpha)$ and $\mathcal{E}(\cdot, \alpha)$ are strictly increasing over $(\alpha,1)$. 

Define $v^E(k,\alpha)\coloneqq V(\mathcal{E}(k,\alpha), L)$ and $v^F(k,\alpha)\coloneqq V(\mathcal{F}(k,\alpha), \min\{\gamma L, H\})$. Note that $v^E(k,\alpha)=u_\Pi(\vx^E)$ and $v^F(k,\alpha)=u_\Pi(\vx^F)$ for a given $(k,\alpha)\in\mathcal L$. By \autoref{lemma:monotonicity}, $V$ is continuous. Since both $\mathcal{E}$ and $\mathcal{F}$ are continuous over $\mathcal L$, so are $v^{E}$ and $v^{F}$. Additionally, by \autoref{lemma:monotonicity}, $v^E$ and $v^F$ inherit the monotonicity properties of $\mathcal{E}$ and $\mathcal{F}$: for each $k\in (0,1)$, the functions $\alpha\mapsto v^E(k,\alpha)$ and $\alpha\mapsto v^F(k,\alpha)$ are strictly decreasing over $(0,k)$ and thus differentiable almost everywhere; for each $\alpha\in (0,1)$, the functions $k\mapsto v^E(k,\alpha)$ and $k\mapsto v^F(k,\alpha)$ are strictly increasing over $(\alpha,1)$ and thus differentiable almost everywhere.

Fix $k'\in (0,1)$. First, consider the trivial case when $v^E(k',\alpha)< v^F(k',\alpha)$ for all $\alpha\leq k'$. Since we have already established that $v^E(k',\alpha)> v^F(k',\alpha)$ for all $\alpha> k'$, the threshold is $\alpha^*(k')=k'$. 

Next, consider the non-trivial case when $v^E(k',\alpha)\geq v^F(k',\alpha)$ for some $\alpha\leq k'$. Since $v^E$ and $v^F$ are continuous over $\mathcal L$, the Intermediate Value Theorem implies there exists $\alpha'\leq k'$ such that $v^E(k',\alpha')= v^F(k',\alpha')$. We now show that this value is unique.  Using $v'\coloneqq v^E(k',\alpha')=v^F(k', \alpha')$ to simplify the notation, we have 
\begin{align*}
    &\frac{\partial }{\partial \alpha}\Big(v^E(k', \alpha')-v^F(k',\alpha')\Big)\\[6pt]
    =&\frac{(1-\rho(v'))\frac{\partial}{\partial \alpha}\mathcal{E}(k',\alpha')}{1+\rho'(v')\big(\mathcal{E}(k', \alpha')-L\big)}- \frac{(1-\rho(v'))\frac{\partial}{\partial \alpha}\mathcal{F}(k',\alpha')}{1+\rho'(v')\big(\mathcal{F}(k', \alpha')-\min\{\gamma L,H\}\big)}\\[6pt]
    \geq &(1-\rho(v'))\frac{\partial}{\partial \alpha}\mathcal{E}(k',\alpha')\left(\frac{1}{1+\rho'(v')\big(\mathcal{E}(k', \alpha')-L\big)}- \frac{1}{1+\rho'(v')\big(\mathcal{F}(k', \alpha')-\min\{\gamma L,H\}\big)}\right)\\[6pt]
    >&0,
\end{align*}
where both inequalities follow because $\rho$ is strictly increasing,  $$\frac{\partial}{\partial \alpha} \mathcal{F}(k',\alpha')\leq \frac{\partial}{\partial \alpha} \mathcal{E}(k',\alpha')<0,$$ and  $\mathcal{E}(k',\alpha')-L>\mathcal{F}(k',\alpha')-\min\{\gamma L,H\}>0$. The final inequality implies that whenever $v^E(k',\cdot)$ crosses $v^F(k', \cdot)$, it does so from below. Consequently, the threshold value $\alpha'$ is unique, so we set $\alpha^*(k')=\alpha'$.

Thus far, we have shown that $\alpha^*$ is well-defined with $\alpha^*(k)\leq k$ for all $k\in(0,1)$. By continuity of $v^E$ and $v^F$ over $\mathcal L$, $\alpha^*$ is continuous. It only remains to show that $\alpha^*$ is strictly increasing on $(0,1)$. Since $\alpha^*$ is continuous and $\alpha^*(k)\leq k$, the set on which $\alpha^*(k)<k$ is a union of disjoint open intervals. At every endpoint of such an interval, continuity implies that $\alpha^*(k)=k$. On the complement of these intervals, $\alpha^*(k)=k$, so strict monotonicity holds immediately. It therefore suffices to prove strict monotonicity on each interval $(\underline k,\bar k)$ over which $\alpha^*(k)<k$. Consider one such open interval and fix $\underline k<k_1<k_2<\bar k$. Since $\alpha^*(k_1)< k_1<k_2$ and $\alpha^*(k_2)<k_2$, we have $(k_1,\alpha^*(k_1))$, $(k_2,\alpha^*(k_1))$, and $(k_2,\alpha^*(k_2))$ all belong to $\mathcal{L}$.

For any $(k,\alpha)\in\mathcal{L}$, $v^E(k,\alpha)-v^F(k,\alpha)\geq 0$ if and only if
\[   G(k,\alpha)\coloneqq \bigl(1-\rho(v^F(k,\alpha))\bigr)\bigl(\mathcal E(k,\alpha)-\mathcal F(k,\alpha)\bigr)-\rho(v^F(k,\alpha))\bigl(\min\{\gamma L,H\}-L\bigr)\geq 0.
\]
As established above, the mapping $k\mapsto v^F(k,\alpha)$ is strictly increasing. Moreover, because
\[
\frac{\partial}{\partial k}\mathcal E(k,\alpha)\leq\frac{\partial}{\partial k}\mathcal F(k,\alpha),
\]
the mapping $k\mapsto \mathcal E(k,\alpha)-\mathcal F(k,\alpha)$ is weakly decreasing. Since $\mathcal E(k,\alpha)-\mathcal F(k,\alpha)>0$, $G(k,\alpha)$ is strictly decreasing in $k$. Similarly, the mapping $\alpha\mapsto v^F(k,\alpha)$ is strictly decreasing. Moreover, because
\[
\frac{\partial}{\partial \alpha}\mathcal E(k,\alpha)\geq\frac{\partial}{\partial \alpha}\mathcal F(k,\alpha),
\]
the mapping $\alpha\mapsto \mathcal E(k,\alpha)-\mathcal F(k,\alpha)$ is weakly increasing. Since $\mathcal E(k,\alpha)-\mathcal F(k,\alpha)>0$, $G(k,\alpha)$ is strictly increasing in $\alpha$. 

Then, by the definition of the cutoffs $\alpha^*(k_1)$ and $\alpha^*(k_2)$, 
\[
G\bigl(k_2,\alpha^*(k_1)\bigr)<G\bigl(k_1,\alpha^*(k_1)\bigr)=0=G\bigl(k_2,\alpha^*(k_2)\bigr).
\]
Since $G(k_2,\cdot)$ is strictly increasing, it follows that $\alpha^*(k_1)<\alpha^*(k_2)$, as desired. Therefore, $\alpha^*$ is strictly increasing. 
\endgroup

\singlespacing
\bibliographystyle{ecta}
\nocite{}\bibliography{fanreferences}

\end{document}